\documentclass[journal,10 pt]{IEEEtran}
\usepackage{bm,cite,algorithm,algorithmic,float,amsmath,amssymb}

\usepackage{cite,graphicx,amsmath,amssymb}
\usepackage{xcolor}
\usepackage{subfigure}
\usepackage{amsthm}
\usepackage{citesort}
\usepackage{fancyhdr}
\usepackage{mdwmath}
\usepackage{mdwtab}
\usepackage{framed}
\usepackage{comment}
\usepackage{graphicx}
\usepackage{epstopdf}

\newtheorem{lemma}{Lemma}
\newtheorem{cor}{Corollary}

\IEEEoverridecommandlockouts

\usepackage{graphicx,epstopdf}
\usepackage{epsfig}	
\usepackage{amsfonts,balance}
\usepackage{bbm}
\floatname{algorithm}{Algorithm}
\usepackage{lipsum}
\usepackage{url}

\newtheorem{theorem}{Theorem}

\newtheorem{corollary}{Corollary}

\makeatletter
\def\ScaleIfNeeded{%
\ifdim\Gin@nat@width>\linewidth \linewidth \else \Gin@nat@width
\fi } \makeatother

\begin{document}

\title{Physical Layer Security in Three-Tier Wireless Sensor Networks: A Stochastic Geometry Approach}

\author{Yansha Deng,~\IEEEmembership{Student Member,~IEEE}, Lifeng Wang,~\IEEEmembership{Student Member,~IEEE}, Maged
Elkashlan,~\IEEEmembership{Member,~IEEE}, Arumugam Nallanathan,~\IEEEmembership{Senior~Member,~IEEE}, and Ranjan K. Mallik,~\IEEEmembership{Fellow,~IEEE}
\thanks{Copyright (c) 2013 IEEE. Personal use of this material is permitted. However, permission to use this material for any other purposes must be obtained from the IEEE by sending a request to pubs-permissions@ieee.org. Manuscript received April 27, 2015; revised Nov 21, 2015; accepted Jan 1, 2016.  This work was supported by the UK Engineering and Physical Sciences Research Council (EPSRC) with Grant No. EP/M016145/1. This paper was presented in part at the IEEE Global
Communications Conference, San Diego, CA, December 2015. The associate editor coordinating the review of this
manuscript and approving it for publication was Prof. T. Charles Clancy.}
\thanks{Y. Deng and A. Nallanathan are with the Department of Informatics, King's College London, London, UK (e-mail:\{yansha.deng, arumugam.nallanathan\}@kcl.ac.uk).}
\thanks{L. Wang is with the Department of Electronic and Electrical Engineering, University College London, London, UK (e-mail:  lifeng.wang@ucl.ac.uk). }
\thanks{ M. Elkashlan is with the School of Electronic Engineering and Computer Science, Queen Mary University of London, London, UK (e-mail:    maged.elkashlan@qmul.ac.uk) }
\thanks{ R. K. Mallik is with the Department of Electrical Engineering, Indian
Institute of Technology - Delhi, Hauz Khas, New Delhi 110016, India (e-mail: rkmallik@ee.iitd.ernet.in).}
}

\maketitle
\setcounter{page}{1} \thispagestyle{plain}

\begin{abstract}
This paper develops a tractable framework for exploiting the potential benefits of physical layer security in three-tier wireless sensor networks using stochastic geometry. In such networks, the sensing data from the remote sensors are collected by sinks with the help of access points, and the external eavesdroppers intercept the data transmissions.  We focus on the secure transmission in two scenarios: i) the active sensors transmit their sensing data to the access points, and ii) the active access points forward the data to the sinks.  We derive new compact expressions for the average secrecy rate in these two scenarios. We also derive a new compact expression for the overall average secrecy rate. Numerical results corroborate our analysis and show that
multiple antennas at the access points can enhance the security of three-tier 
wireless sensor networks.
Our results show that increasing the number of access points  decreases the average secrecy rate between the access point and its associated sink. However, we find that increasing the number of access points first increases the overall average secrecy rate, with a critical value beyond which the overall average secrecy rate then decreases.
 When increasing the number of active sensors, both the average secrecy rate between the sensor and its associated access point and the overall average secrecy rate decrease. In contrast, increasing the number of sinks improves both the average secrecy rate between the access point and its associated sink, as well as the overall average secrecy rate.


\end{abstract}


\section{Introduction}
Due to its wide applications such as environmental sensing, health monitoring, and military communications~\cite{Akyildiz2002}, wireless sensor networks (WSNs) have attracted considerable attention from the industry and academia. The security of WSNs is a big concern, since the broadcast nature of wireless channels is susceptible to eavesdropping and the sensing data needs to be protected. In practice, the small-size, low-cost and low-power sensors are randomly deployed to sense the data, which is sent back to the sinks by multihop transmissions. Multihop architectures pose great challenges to conventional cryptographic methods involving key distribution and management, and result in high complexity in data encryption and decryption. Physical layer security has emerged as an appealing low-complexity approach to secure the information transmission. The core idea behind it is to exploit the characteristics of wireless channels such as fading or noise to transmit a message from a source to an
intended destination while keeping the message confidential from eavesdroppers. Motivated by this, the potential applications of physical layer security have been investigated in various wireless networks such as cellular networks, cognitive radio, ad-hoc, etc.

\subsection{Physical Layer Security:  Current State-of-the-Art}
In the 1970s, Aaron D. Wyner first introduced physical layer security~\cite{Wyner}. Triggered by the rapid evolution of wireless network architectures, the idea of enabling security at physical layer has drawn the attention of the wireless community~\cite{Poor2012}. In cellular networks, physical layer security is important
for adding an extra level of protection~\cite{Geraci_downlink,HeWang_2013}. In~\cite{Geraci_downlink}, secure downlink transmission in cellular networks was investigated, and the secrecy using linear precoding based on regularized channel inversion was examined. In multi-cell environments, the cell association and location information of mobile users play an important role in secrecy performance~\cite{HeWang_2013}. Although it can alleviate the scarcity of radio frequency spectrum, security of cognitive radio networks is critical as it is easily exposed to external threats~\cite{Pei:10:TWC}. In~\cite{Pei:10:TWC}, the optimal secrecy beamforming in a multiple-input single-output (MISO) cognitive radio wiretap channel was proposed. 
  In cooperative networks, relays are deployed to boost the coverage and reliability, however, the relay can be trusted~\cite{LunDong,Yulong_JSAC} or untrusted~\cite{XiangHe2008,Lifeng_WCL} where the untrusted relay is thought of as an eavesdropper. In~\cite{LunDong}, the design of trusted relay weights and allocation of transmit power under different relay protocols such as amplify-and-forward (AF), decode-and-forward (DF), and cooperative jamming (CJ) was considered. In~\cite{Yulong_JSAC}, trusted relay selection schemes based on the AF and DF protocols were proposed to improve physical layer security. In untrusted relay networks, CJ was introduced to confuse the untrusted relay~\cite{XiangHe2008}. Joint power allocation and CJ was developed in~\cite{Lifeng_WCL}, and it was shown that a positive secrecy rate can be guaranteed. In decentralized networks such as ad-hoc, the public-key cryptography is expensive and difficult~\cite{Zhou:TWC:2011_feb,Xiangyun_2011_Aug,XiZhang_2013_TIFS}. In \cite{Zhou:TWC:2011_feb}, the secure connectivity in wireless random networks was studied, and the eigen-beamforming was implemented to maximize the signal strength to the intended receiver. In \cite{Xiangyun_2011_Aug}, the secrecy transmission capacity in wireless ad-hoc networks was analyzed, and the secrecy guard zone was introduced to improve the secrecy transmission capacity. In~\cite{XiZhang_2013_TIFS}, the transmit beamforming with artificial noise strategies were used to enhance the secrecy in large-scale ad-hoc networks.

Physical layer security schemes have been recently proposed for WSNs to combat eavesdropping~\cite{Xiaohua2005,Marano2009,Soosahabi2012,Barcelo_Llado2014}. In \cite{Xiaohua2005}, the downlink secure transmission from the mobile agent to the authorized user was considered and perfect secrecy can be achieved by intentionally creating channel variation. In~\cite{Marano2009}, a detection problem under physical layer secrecy
constraints in an energy-constrained WSNs was addressed, and the optimal operative solutions were analyzed. In~\cite{Soosahabi2012}, sensor transmissions were observed by the authorized fusion center (FC) and unauthorized (third party) FC. It was shown that physical layer security for distributed detection is scalable due to its low computational complexity. More recently in~\cite{Barcelo_Llado2014}, compressed sensing (CS) was introduced to provide secrecy against eavesdropping in addition to the other CS benefits.

\subsection{Approach and Contributions}
In this paper, we examine the potential benefits of  physical layer security in a three-tier WSN using stochastic geometry modeling. In three-tier WSNs, the sensors are located far from the sinks, and the access points are deployed to help the sensors forward their data to the sinks.  Confidential information transmissions are intercepted by the eavesdroppers. Considering the fact that sensors are densely deployed and their locations are randomly distributed~\cite{Akyildiz2002}, we introduce stochastic geometry to model the locations of the nodes in WSNs. Such a modeling approach has been applied in heterogeneous networks~\cite{Dhillon2012} and cognitive radio networks~\cite{Chia_han2012}. Our main contributions are summarized as follows.

\begin{itemize}
  \item We develop a new analytical framework to examine the implementation of physical layer security in three-tier WSNs. The locations and spatial densities of sensors, access points, sinks, and eavesdroppers are modeled using stochastic geometry. Each access point is equipped with multiple antennas and uses the low-complexity maximal-ratio combining (MRC) to receive the  data signals from the sensors and maximal-ratio transmission (MRT) beamformer to transmit the signals. We investigate the secure transmissions between the active sensors and access points, and beween the active access points and sinks.

  \item We present new statistical properties, based on which we derive new compact expressions for the average secrecy rate between the typical sensor and its associated access point, and between the typical access point and its associated sink.  We also derive the minimum number of sinks required for a target average secrecy rate. Particularly, we derive a new compact expression for overall average secrecy rate in three-tier WSNs.

  \item We show that using MRC/MRT at access points can enhance the secure transmission. Based on the proposed analysis and simulations, several important observations are reached: 1) the average secrecy rate decreases as the number of sensors grows large, due to more interference from sensors, 2) the average secrecy rate increases with increasing the number of sinks, because of the shorter distances between the access points and their associated sinks, and 3) the overall average secrecy rate increases with increasing the number of access points. However, beyond a critical value, the overall average secrecy rate decreases with increasing the number of access points.

\end{itemize}

The notation of this paper is given in Table~\ref{Tab1}.
\begin{table}[tbp]\label{Tab1}
\centering
\caption{Notation}
\begin{tabular}{|c|l|}
\hline
{$\Phi_{s,a}$} & Poison point process (PPP) of sensor locations\\ 
$\lambda_s$ & Intensity of $\Phi_s$ \\ 
$\Phi_{ap,a}$ & PPP of access points locations\\ 
$\lambda_{ap}$ & Intensity of $\Phi_{ap}$ \\ 
$\Phi_{sk}$	& PPP of sinks locations \\ 
$\lambda_{sk}$	& Intensity of $\Phi_{sk}$ \\ 
$\rho_s$	& The probability that sensor is triggered to transmit \\ & the data\\ 
$\rho_{ap}$	& The activity probability of access point that forwards \\ & the data to the sinks \\ 
$\Phi_{s,e}$ & PPP of eavesdropper locations, where the eavesdroppers \\ & intercept the sensors' data \\ 
$\Phi_{ap,e}$ & PPP of eavesdropper locations, where the eavesdroppers \\ & intercept the access points' data  \\ 
$\lambda_e^s$ & Intensity of $\Phi_{s,e}$  \\ 
$\lambda_e^{ap}$ & Intensity of $\Phi_{ap,e}$  \\ 
$\dag$   & Conjugate transpose \\
\hline
\end{tabular}
\label{table:2}
\end{table}

\section{System Description}

\begin{figure}[t!]
    \begin{center}
        \includegraphics[width=2.8 in]{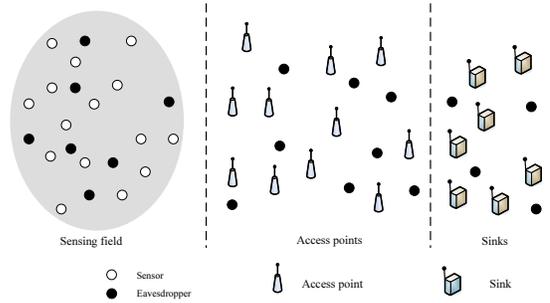}
        \caption{Illustration of three-tier wireless sensor networks, where the sensors transmit the data to the sinks via the access points, in the presence of eavesdropping.}
        \label{wireless_sensor_networks}
    \end{center}
\end{figure}
As shown in Fig. \ref{wireless_sensor_networks}, a three-tier WSN is considered, where the geographically remote sensors transmit the sensed data to the sinks with the help of half-duplex decode-and-forward (DF) access points with no direct links between sensors and sinks. The eavesdroppers overhears the data transmission without modifying it. In the sensing field, sensors are randomly located according to a homogeneous Poisson point process (HPPP) $\Phi_{s}$ with intensity $\lambda_{s}$. 
In order to consider unplanned deployment of the access points and sinks, the random locations of 
the access points and sinks are approximated  as independent HPPPs $\Phi_{ap}$ and $\Phi_{sk}$ with intensities $\lambda_{ap}$ and $\lambda_{sk}$, respectively, which is suitable in large scale networks \cite{taesoo2013random}. Since the sensors may transmit data intermittently, the activity probability of a sensor that is triggered to transmit the data is denoted as $\rho_s$ ($0<\rho_s<1$), and the activity probability of an access point that forwards the data to the sink is denoted as $\rho_{ap}$ ($0<\rho_{ap}<1$)\footnote{{In  practical scenarios, the access points operate in three modes: receiving the sensed data from active sensors, forwarding the sensed data to the sinks, and idle. The activity probability of sensor only determines the
number of access points which receive the data from the active sensors, and is independent of the number of access points which forward the data to the sink. The number of active access points that are triggered to forward the sensed data to sinks depends on the availability of sinks. 
As such,  $\rho_s$ and $\rho_{ap}$ are independent values.}}.  
We assume that the probability of being an active sensor/access point is independent of the access point/sink's location. Therefore,  the active sensors  and active access points  constitute  independent HPPPs  $\Phi_{s,a}$ and $\Phi_{ap,a}$ with intensities $\lambda_{s}\rho_s$ and $\lambda_{ap}\rho_{ap}$, respectively \cite{taesoo2013random}.
Non-colluding eavesdroppers are considered and eavesdroppers' locations are modeled as two independent HPPPs $\Phi_{s,e}$ and $\Phi_{ap,e}$ with intensities $\lambda_e^s$ and $\lambda_e^{ap}$, respectively.  The eavesdroppers in $\Phi_{s,e}$ intercept the data transmitted by the sensors and the eavesdroppers in $\Phi_{ap,e}$ intercept the data transmitted by the access points. Note that  the eavesdroppers in $\Phi_{s,e}$ and in $\Phi_{ap,e}$ are far from each other.

In this three-tier network, the sensor is associated with its nearest access point to receive the sensor's data and the access point is associated with its nearest sink to receive the access point's data \footnote{{In reality, there may be more than one active sensor/access point to choose the same access point/sink; this can be effectively dealt with using multiple access techniques.}}.  Each access point is equipped with $M$ antennas, and the sensors and sinks are single-antenna nodes. To enhance the information transmission, the access points use MRC to receive the sensors' data signals and MRT beamformer to transmit the signals.
 The wireless channels are modeled as independent quasi-static Rayleigh fading.

An arbitrary typical sensor $o$ transmits data to its nearest access point (called typical access point). The typical access point not only receives the useful data from the typical sensor, but is also subject to the interference from  other active sensors  and   active access points. Thus, the receive signal-to-interference-plus-noise ratio (SINR) after MRC at its corresponding typical access point is given by
\begin{align}\label{SINR_AP}
{\gamma _{ap}} = \frac{{{{\left\| {{{\bf{h}}_{{s_0},a{p_0}}}} \right\|}^2}{{\left| {{X_{{s_0},a{p_0}}}} \right|}^{ - \alpha }}}}{{\underbrace {{I_{s,ap}} + {{I_{ap,ap}}}}_{I{n_{ap}}} + {{{\delta ^2}} \mathord{\left/
 {\vphantom {{{\delta ^2}} {{P_s}}}} \right.
 \kern-\nulldelimiterspace} {{P_s}}}}},
\end{align}
where ${I_{s,ap}} = {\sum _{i \in {\Phi _{s,a}}\backslash \left\{ {{s_0}} \right\}}}{\left| {\frac{{{{\bf{h}}_{{s_0},a{p_0}}}^\dag }}{{\left\| {{{\bf{h}}_{{s_0},a{p_0}}}} \right\|}}{{\bf{h}}_{i,a{p_0}}}} \right|^2}{\left| {{X_{i,a{p_0}}}} \right|^{ - \alpha }}$, ${I_{ap,ap}} = \mu{\sum _{j \in {\Phi _{ap,a}}\backslash \left\{ {a{p_0}} \right\}}}{\left| {\frac{{{{\bf{h}}_{{s_0},a{p_0}}}^\dag }}{{\left\| {{{\bf{h}}_{{s_0},a{p_0}}}} \right\|}}{{\bf{H}}_{j,a{p_0}}}\frac{{{{\bf{h}}_{j,s{k_j}}}^\dag }}{{\left\| {{{\bf{h}}_{j,s{k_j}}}} \right\|}}} \right|^2}{\left| {{X_{j,a{p_0}}}} \right|^{ - \alpha }}$,
and $\mu  = {{{P_{ap}}} \mathord{\left/
 {\vphantom {{{P_{ap}}} {{P_s}}}} \right.
 \kern-\nulldelimiterspace} {{P_s}}}$. Note that the interfering access points deliver their own data to their corresponding sinks using MRT beamformer vector  $\frac{{{{\bf{h}}_{j,s{k_j}}}^\dag }}{{\left\| {{{\bf{h}}_{j,s{k_j}}}} \right\|}}$,  which are also received and combined at the typical access point with MRC vector $\frac{{{{\bf{h}}_{{s_0},a{p_0}}}^\dag }}{{\left\| {{{\bf{h}}_{{s_0},a{p_0}}}} \right\|}}$.
Here,    ${{{\bf{h}}_{{s_0},a{p_0}}}}$ and ${\left| {{X_{{s_0},a{p_0}}}} \right|}$ are the channel fading vector and distance between the typical sensor and its typical access point, respectively,  $\alpha$ is the path loss exponent, ${{{\bf{h}}_{i,a{p_0}}}}\in\mathcal{C}^{M \times 1}$ and ${\left| {{X_{i,a{p_0}}}} \right|}$ are the channel fading vector and distance between the sensor $i$ and the typical access point, respectively,  ${{\bf{H}}_{j,a{p_0}}}$ and ${\left| {{X_{j,a{p_0}}}} \right|}$ are the channel fading matrix and distance between the interfering access point $j$ and the typical access point, respectively, ${\bf{h}}_{j,s{k_j}} \in\mathcal{C}^{1 \times M}$ is the channel fading vector between the interfering access point $j$ and its corresponding sink, $P_s$ is the sensor's transmit power, $P_{ap}$ is the access point's transmit power, and ${\delta^2}$ is the noise power.

We consider the non-colluding eavesdropping scenario, in which the most detrimental eavesdropper that has the highest receive SINR dominates the secrecy rate~\cite{LunDong}. {An arbitrary eavesdropper $e_k$ that intercepts the sensor and the access point transmission overhears
 the useful signal from the typical sensor to the typical access point, and simultaneously receives the interfering  data from the other  active sensors and  active access points.   This eavesdropper suffers from the interfering signals emitted by the other interfering access points using the MRT beamformer 
$\frac{{{{\bf{h}}_{j,s{k_k}}}^\dag }}{{\left\| {{{\bf{h}}_{j,s{k_k}}}} \right\|}}$.} Thus, the received SINR  at the most detrimental eavesdropper in $\Phi_{s,e}$ for the sensor and the access point transmission  is given by
\begin{align}\label{SINR_Sensor_Eve}
{\gamma _{s,e}} = \mathop {\max }\limits_{{e_k} \in {\Phi _{s,e}}} \left\{ {\frac{{{{\left| {{h_{{s_0},{e_k}}}} \right|}^2}{{\left| {{X_{{s_0},{e_k}}}} \right|}^{ - \alpha }}}}{{\underbrace {{I_{s,e}} + {I_{ap,e}}}_{I{n_{s,e}}} + {{{\delta ^2}} \mathord{\left/
 {\vphantom {{{\delta ^2}} {{P_s}}}} \right.
 \kern-\nulldelimiterspace} {{P_s}}}}}} \right\},
\end{align}
where  ${I_{s,e}} = {\sum _{i \in {\Phi _{s,a}}\backslash \left\{ {{s_0}} \right\}}}{\left| {{h_{i,{e_k}}}} \right|^2}{\left| {{X_{i,{e_k}}}} \right|^{ - \alpha }}$ and ${I_{ap,e}} = {\sum _{j \in {\Phi _{ap,a}}\backslash \left\{ {a{p_0}} \right\}}}\mu {\left| {{{\bf{h}}_{j,{e_k}}}\frac{{{{\bf{h}}_{j,s{k_j}}}^\dag }}{{\left\| {{{\bf{h}}_{j,s{k_j}}}} \right\|}}} \right|^2}{\left| {{X_{j,{e_k}}}} \right|^{ - \alpha }}$,
${{h_{{s_0},{e_k}}}}$ and ${\left| {{X_{{s_0},{e_k}}}} \right|}$ are the channel fading coefficient and distance between the typical sensor and the eavesdropper, respectively, ${{h_{i,{e_k}}}}$ and ${\left| {{X_{i,{e_k}}}} \right|}$ are the channel fading coefficient and distance between sensor $i$ and the eavesdropper, respectively, and ${{{\bf{h}}_{j,{e_k}}}}$ and ${\left| {{X_{j,{e_k}}}} \right|}$ are the channel fading vector and distance between the access point $j$ and the eavesdropper, respectively.

 After receiving the typical  sensor's data, the typical  access point ${ap}_0$ will forward the sensed data to the nearest sink (called typical sink) $sk_0$ for data collection. { Due to the current transmission from other active access points, the typical sink suffers from their interferences. } As such, the received SINR at the typical sink $sk_0$ is given by
\begin{align}\label{SINR_sink}
{\gamma _{sk}} = \frac{{{{\left\| {{{\bf{g}}_{a{p_0},s{k_0}}}} \right\|}^2}{{\left| {{X_{a{p_0},s{k_0}}}} \right|}^{ - \beta }}}}{{I{n_{ap,sk}} + {{{\delta ^2}} \mathord{\left/
 {\vphantom {{{\delta ^2}} {{P_{ap}}}}} \right.
 \kern-\nulldelimiterspace} {{P_{ap}}}}}},
\end{align}
where $I{n_{ap,sk}} = {\sum _{j \in {\Phi _{ap,a}}\backslash \left\{ {a{p_0}} \right\}}}{\left| {{{\bf{g}}_{j,s{k_0}}}\frac{{{{\bf{h}}_{j,s{k_j}}}^\dag }}{{\left\| {{{\bf{h}}_{j,s{k_j}}}} \right\|}}} \right|^2}{\left| {{X_{j,s{k_0}}}} \right|^{ - \beta }}$, ${{{\bf{g}}_{a{p_0},s{k_0}}}}\in\mathcal{C}^{1 \times M}$ and ${\left| {{X_{a{p_0},s{k_0}}}} \right|}$ are the channel fading vector and distance between the typical access point and its typical sink, respectively, $\beta$ is the path loss exponent, ${{{\bf{g}}_{j,s{k_0}}}}\in\mathcal{C}^{1 \times M}$ and ${\left| {{X_{j,s{k_0}}}} \right|}$ are the channel fading vector and distance between the access point $j$ and the typical sink, and ${{{\bf{h}}_{j,s{k_j}}}}\in\mathcal{C}^{1 \times M}$ is the channel fading vector between the access point $j$ and its associated sink. 

 {An arbitrary eavesdropper $e_k$ that intercepts the typical access point and the typical sink transmission overhears
 the signal transmitted by the typical access point with the MRT beamformer $\frac{{{{\bf{g}}_{a{p_0},s{k_0}}}^\dag }}{{\left\| {{{\bf{g}}_{a{p_0},S{k_0}}}} \right\|}}$, and suffers from the interfering signals emitted by  other interfering access points with the MRT beamformer 
$\frac{{{{\bf{h}}_{j,s{k_k}}}^\dag }}{{\left\| {{{\bf{h}}_{j,s{k_k}}}} \right\|}}$.} Thus,
the received SINR   at the most detrimental eavesdropper  for the access point and the sink transmission is given by
 \begin{align}\label{SINR_AP_EVE}
 {\gamma _{ap,{e}}} = \mathop {\max }\limits_{{e_k} \in {\Phi _{ap,e}}} \Bigg\{ {\frac{{{{\left| {{{\bf{g}}_{a{p_0},{e_k}}}\frac{{{{\bf{g}}_{a{p_0},s{k_0}}}^\dag }}{{\left\| {{{\bf{g}}_{a{p_0},S{k_0}}}} \right\|}}} \right|}^2}{{\left| {{X_{a{p_0},{e_k}}}} \right|}^{ - \beta }}}}{{I{n_{ap,e}} + {{{\sigma ^2}} \mathord{\left/
 {\vphantom {{{\sigma ^2}} {{P_{ap}}}}} \right.
 \kern-\nulldelimiterspace} {{P_{ap}}}}}}} \Bigg\},
 \end{align}
where $I{n_{ap,e}} = {\sum _{j \in {\Phi _{ap,a}}\backslash \left\{ {a{p_0}} \right\}}}{\left| {{{\bf{g}}_{j,{e_k}}}\frac{{{{\bf{h}}_{j,s{k_k}}}^\dag }}{{\left\| {{{\bf{h}}_{j,s{k_k}}}} \right\|}}} \right|^2}{\left| {{X_{j,{e_k}}}} \right|^{ - \beta }}$, ${{g_{a{p_0},{e_k}}}}$ and ${\left| {{X_{a{p_0},{e_k}}}} \right|}$ are the channel fading coefficient and distance between the typical access point and the eavesdropper, respectively, and ${{{\bf{g}}_{j,{e_k}}}}$ and ${\left| {{X_{j,{e_k}}}} \right|}$ are the channel fading vector and distance between the access point $j$ and the eavesdropper, respectively.

\section{Secrecy Performance Evaluations}
In this section, we characterize the secrecy performance in terms of average secrecy rate. Before exhibiting the overall secrecy performance behaviors, we evaluate the secrecy of two different links, namely the link between the sensor and access point, and the link between the access point and sink. We derive new analytical expressions for the average secrecy rate, and analyze the impact of the two links on the overall average secrecy rate.

\subsection{Average Secrecy Rate between Sensor and Access Point}
 {We evaluate the average secrecy rate based on the worst-case, where the eavesdropper with the best SINR is used to calculate the  average secrecy rate~\cite{LunDong}.} Hence, for a typical link between a typical sensor and its associated access point, the instantaneous secrecy rate is defined as~\cite{Yuksel2011}
\begin{align}\label{secrecy_rate_Eq_s}
{C_s^{ap}} =[ { {{C_{ap}} - {C_{s,e}}}  }]^+,
\end{align}
where $[x]^+=\mathrm{max}\{x,0\}$, $C_{ap}=\log_2\left(1+{\gamma _{ap}}\right)$ is the capacity of the channel between the typical sensor and access point, and $C_{s,e}=\log_2\left(1+{\gamma _{s,e}}\right)$ is the capacity of the eavesdropping channel between the typical sensor and the most detrimental eavesdropper.

 \subsubsection{New Statistics}
We derive the cumulative distribution functions (CDFs) of SINRs at the typical access point and the most detrimental eavesdropper that intercepts the transmission between the typical sensor and the  access point  in 
{\bf{Lemma \ref{Lemma1}}}   and {\bf{Lemma \ref{Lemma2}}}, respectively.

\begin{lemma}\label{Lemma1}
The CDF of SINR at the typical access point is derived as \eqref{CDFAP} at the top of next page.
\begin{figure*}[t]
\normalsize
\begin{align}\label{CDFAP}
&{F_{{\gamma _{ap}}}}\left( {{\gamma _{th}}} \right) = 1 - 2\pi {\lambda _{ap}}\left( {1 - {\rho _{ap}}} \right)\int_0^\infty  {r\exp \Big\{ { - \left( {{\lambda _s}{\rho _s} + {\lambda _{ap}}{\rho _{ap}}{\mu ^{\frac{2}{\alpha }}}} \right)} } \pi \Gamma \left( {1 + {2 \mathord{\left/
 {\vphantom {2 \alpha }} \right.
 \kern-\nulldelimiterspace} \alpha }} \right) {\Gamma \left( {1 - {2 \mathord{\left/
 {\vphantom {2 \alpha }} \right.
 \kern-\nulldelimiterspace} \alpha }} \right){{\left( {{\gamma _{th}}} \right)}^{\frac{2}{\alpha }}}{r^2}}{ - {{{\gamma _{th}}{r^\alpha }{\delta ^2}} \mathord{\left/
 {\vphantom {{{\gamma _{th}}{r^\alpha }{\delta ^2}} {{P_s}}}} \right.
 \kern-\nulldelimiterspace} {{P_s}}}}\nonumber\\
 &{ - \pi {\lambda _{ap}}\left( {1 - {\rho _{ap}}} \right){r^2}}\Big\}dr - 2\pi {\lambda _{ap}}\left( {1 - {\rho _{ap}}} \right)\sum\limits_{m = 1}^{M - 1} {\frac{{{{\left( {{r^\alpha }} \right)}^m}}}{{{{\left( { - 1} \right)}^m}}}\sum {\frac{1}{{\prod\limits_{l = 1}^m {{m_l}!l{!^{{m_l}}}} }}} } \int_0^\infty  {r\exp \Big\{ { - \left( {{\lambda _s}{\rho _s} + {\lambda _{ap}}{\rho _{ap}}{\mu ^{\frac{2}{\alpha }}}} \right)\pi \Gamma \left( {1 + {2 \mathord{\left/
 {\vphantom {2 \alpha }} \right.
 \kern-\nulldelimiterspace} \alpha }} \right)} }\nonumber\\
&{\Gamma \left( {1 - {2 \mathord{\left/
 {\vphantom {2 \alpha }} \right.
 \kern-\nulldelimiterspace} \alpha }} \right){{\left( {{\gamma _{th}}} \right)}^{\frac{2}{\alpha }}}{r^2} - {{{\gamma _{th}}{r^\alpha }{\delta ^2}} \mathord{\left/
 {\vphantom {{{\gamma _{th}}{r^\alpha }{\delta ^2}} {{P_s}}}} \right.
 \kern-\nulldelimiterspace} {{P_s}}} - \pi {\lambda _{ap}}\left( {1 - {\rho _{ap}}} \right){r^2}} \Big\} \Big[ { - {2 \mathord{\left/
 {\vphantom {2 \alpha }} \right.
 \kern-\nulldelimiterspace} \alpha }\left( {{\lambda _s}{\rho _s} + {\lambda _{ap}}{\rho _{ap}}{\mu ^{\frac{2}{\alpha }}}} \right)} {\pi \Gamma \left( {1 + {2 \mathord{\left/
 {\vphantom {2 \alpha }} \right.
 \kern-\nulldelimiterspace} \alpha }} \right)\Gamma \left( {1 - {2 \mathord{\left/
 {\vphantom {2 \alpha }} \right.
 \kern-\nulldelimiterspace} \alpha }} \right){{\left( {{\gamma _{th}}} \right)}^{{2 \mathord{\left/
 {\vphantom {2 \alpha }} \right.
 \kern-\nulldelimiterspace} \alpha }}}}\nonumber\\
 &{ {r^{\left( {2 - \alpha } \right)}}- {{{\gamma _{th}}{\delta ^2}} \mathord{\left/
 {\vphantom {{{\gamma _{th}}{\delta ^2}} {{P_s}}}} \right.
 \kern-\nulldelimiterspace} {{P_s}}}}\Big]^{m_1}\prod\limits_{l = 2}^m {\Big[ { - \left( {{\lambda _s}{\rho _s} + {\lambda _{ap}}{\rho _{ap}}{\mu ^{\frac{2}{\alpha }}}} \right)\pi \Gamma \left( {1 + {2 \mathord{\left/
 {\vphantom {2 \alpha }} \right.
 \kern-\nulldelimiterspace} \alpha }} \right)\Gamma \left( {1 - {2 \mathord{\left/
 {\vphantom {2 \alpha }} \right.
 \kern-\nulldelimiterspace} \alpha }} \right)}} {{{\left( {{\gamma _{th}}} \right)}^{{2 \mathord{\left/
 {\vphantom {2 \alpha }} \right.
 \kern-\nulldelimiterspace} \alpha }}}\prod\limits_{j = 0}^{l - 1} {\left( {{2 \mathord{\left/
 {\vphantom {2 \alpha }} \right.
 \kern-\nulldelimiterspace} \alpha } - j} \right){r^{2 - l\alpha }}} }\Big]^{m_l}dr,\\
 &\mathrm{where} \sum\limits_{l{\rm{ = }}1}^m {l\cdot{m_l}}=m. \nonumber
\end{align}

\hrulefill \vspace*{0pt}
\end{figure*}

\end{lemma}

\begin{proof}
See Appendix A.
\end{proof}

\begin{lemma}\label{Lemma2}
The CDF of SINR at the most detrimental eavesdropper which intercepts the transmission between the typical sensor and the  access point is derived as
\begin{align}\label{CDFAP_E}
&{F_{{\gamma _{s,e}}}}\left( \gamma _{th}\right) = \nonumber \\
&\exp \left\{ { - \pi \lambda _e^s\int_0^\infty  {\exp \Big\{ { - \left( {{\lambda _s}{\rho _s} + {\lambda _{ap}}{\rho _{ap}}{\mu ^{{2 \mathord{\left/
 {\vphantom {2 \alpha }} \right.
 \kern-\nulldelimiterspace} \alpha }}}} \right)\pi } } } \right.
\nonumber \\
&\ {\Gamma \left( {1 + {2 \mathord{\left/
 {\vphantom {2 \alpha }} \right.
 \kern-\nulldelimiterspace} \alpha }} \right)\Gamma \left( {1 - {2 \mathord{\left/
 {\vphantom {2 \alpha }} \right.
 \kern-\nulldelimiterspace} \alpha }} \right){{\left( {{\gamma _{th}}} \right)}^{\frac{2}{\alpha }}}t - {{{\delta ^2}{\gamma _{th}}{t^{{\alpha  \mathord{\left/
 {\vphantom {\alpha  2}} \right.
 \kern-\nulldelimiterspace} 2}}}} \mathord{\left/
 {\vphantom {{{\delta ^2}{\gamma _{th}}{t^{{\alpha  \mathord{\left/
 {\vphantom {\alpha  2}} \right.
 \kern-\nulldelimiterspace} 2}}}} {{P_s}}}} \right.
 \kern-\nulldelimiterspace} {{P_s}}}} \Big\}dt.
\end{align}

\end{lemma}

\begin{proof}
See Appendix B.
\end{proof}

\subsubsection{Average Secrecy Rate}
Based on our fundamental work in~\cite{lifeng2014physical}, the average secrecy rate between the sensor and the access point is the average of secrecy rate ${{ C}_{s}^{ap}}$ over $\gamma_{s,e}$  and $\gamma_{ap}$, which can be written as
\begin{align}\label{first_Hop_ASR}
{{\bar C}_{s}^{ap}} &= \frac{1}{{\ln 2}}\int_0^\infty  {\frac{{{F_{{\gamma _{s,e}}}}\left( x \right)}}{{1 + x}}(1 - {F_{{\gamma _{ap}}}}\left( x \right))} dx.
\end{align}
By substituting the CDF of ${\gamma _{ap}}$ in \eqref{CDFAP} and the CDF of  ${\gamma _{s,e}}$ in \eqref{CDFAP_E} into \eqref{first_Hop_ASR},
 we can obtain the average secrecy rate between the sensor and the access point.

Note that the derived average secrecy rate between the sensor and the access point in \eqref{first_Hop_ASR} is not in a simple form. As such, in the following corollary,  we present the interference-limited case for the  average secrecy rate with a single antenna at the access point.
\begin{cor}
When the access points are equipped with single antenna in the interference-limited scenario, the  average secrecy rate between the sensor and the access point
is given by
\begin{align}\label{first_Hop_ASR_special}
&{{\bar C}_{s}^{ap}} = \nonumber \\& \frac{{\pi {\lambda _{ap}}\left( {1 - {\rho _{ap}}} \right)}}{{\ln 2}}\int_0^\infty  {\frac{{\exp \left\{ { - {{\pi \lambda _e^s} \mathord{\left/
 {\vphantom {{\pi \lambda _e^s} {\left( {{\Lambda _1}{x^{{2 \mathord{\left/
 {\vphantom {2 \alpha }} \right.
 \kern-\nulldelimiterspace} \alpha }}}} \right)}}} \right.
 \kern-\nulldelimiterspace} {\left( {{\Lambda _1}{x^{{2 \mathord{\left/
 {\vphantom {2 \alpha }} \right.
 \kern-\nulldelimiterspace} \alpha }}}} \right)}}} \right\}}}{{\left( {1 + x} \right)\left( {{\Lambda _1}{x^{{2 \mathord{\left/
 {\vphantom {2 \alpha }} \right.
 \kern-\nulldelimiterspace} \alpha }}} + \pi {\lambda _{ap}}\left( {1 - {\rho _{ap}}} \right)} \right)}}dx} ,
\end{align}
where ${\Lambda _1} = \left( {{\lambda _s}{\rho _s} + {\lambda _{ap}}{\rho _{ap}}{\mu ^{\frac{2}{\alpha }}}} \right)\pi \Gamma \left( {1 + {2 \mathord{\left/
 {\vphantom {2 \alpha }} \right.
 \kern-\nulldelimiterspace} \alpha }} \right)\Gamma \left( {1 - {2 \mathord{\left/
 {\vphantom {2 \alpha }} \right.
 \kern-\nulldelimiterspace} \alpha }} \right).$
 \end{cor}

\subsection{Average Secrecy Rate between Access Point and Sink}

Similar to \eqref{secrecy_rate_Eq_s}, for   a typical access point and its associated sink, the instantaneous secrecy rate is defined as
\begin{align}\label{ASR_AccessPoint_Sink}
{C_s^{sk}} =[ { {{C_{sk}} - {C_{ap,e}}}  }]^+,
\end{align}
where $C_{sk}=\log_2\left(1+{\gamma _{sk}}\right)$ and $C_{ap,e}=\log_2\left(1+{\gamma _{ap.e}}\right)$.
 \subsubsection{New Statistics}
We derive the CDFs of SINRs at the typical sink and the most detrimental eavesdropper that intercepts the  transmission between the typical access point and the sink   in 
{\bf{Lemma \ref{Lemma3}}}   and {\bf{Lemma \ref{Lemma4}}}, respectively.

\begin{lemma}\label{Lemma3}
The CDF of SINR at the typical sink is derived as
\begin{align}\label{CDFSK_Sink_2hop}
&{F_{{\gamma _{sk}}}}\left( x \right) =1 - 2\pi {\lambda _{sk}}\int_0^\infty  {r\exp \Big\{ { - {\lambda _{ap}}{\rho _{ap}}\pi \Gamma \left( {1 + {2 \mathord{\left/
 {\vphantom {2 \beta }} \right.
 \kern-\nulldelimiterspace} \beta }} \right)}} \nonumber\\
& {\Gamma \left( {1 - {2 \mathord{\left/
 {\vphantom {2 \beta }} \right.
 \kern-\nulldelimiterspace} \beta }} \right){{\left( {{\gamma _{th}}} \right)}^{\frac{2}{\beta }}}{r^2} - {{{\gamma _{th}}{r^\beta }{\delta ^2}} \mathord{\left/
 {\vphantom {{{\gamma _{th}}{r^\beta }{\delta ^2}} {{P_{ap}}}}} \right.
 \kern-\nulldelimiterspace} {{P_{ap}}}} - \pi {\lambda _{sk}}{r^2}} \Big\}dr  - 2\pi {\lambda _{sk}}\nonumber\\
&\sum\limits_{m = 1}^{M - 1} {\frac{1}{{{{\left( { - 1} \right)}^m}}}\sum {\frac{1}{{\prod\limits_{l = 1}^m {{m_l}!l{!^{{m_l}}}} }}} } \int_0^\infty  {{r^{\beta m + 1}}\exp \Big\{ { - {\lambda _{ap}}{\rho _{ap}}\pi } } \nonumber\\
& {\Gamma \left( {1 + {2 \mathord{\left/
 {\vphantom {2 \beta }} \right.
 \kern-\nulldelimiterspace} \beta }} \right)\Gamma \left( {1 - {2 \mathord{\left/
 {\vphantom {2 \beta }} \right.
 \kern-\nulldelimiterspace} \beta }} \right){{\left( {{\gamma _{th}}} \right)}^{\frac{2}{\beta }}}{r^2} - {{{\gamma _{th}}{r^\beta }{\delta ^2}} \mathord{\left/
 {\vphantom {{{\gamma _{th}}{r^\beta }{\delta ^2}} {{P_{ap}}}}} \right.
 \kern-\nulldelimiterspace} {{P_{ap}}}} - \pi {\lambda _{sk}}{r^2}} \Big\}\nonumber\\
&\Big[  - {\lambda _{ap}}{\rho _{ap}}\pi \frac{2}{\beta }\Gamma \left( {1 + {2 \mathord{\left/
 {\vphantom {2 \beta }} \right.
 \kern-\nulldelimiterspace} \beta }} \right)\Gamma \left( {1 - {2 \mathord{\left/
 {\vphantom {2 \beta }} \right.
 \kern-\nulldelimiterspace} \beta }} \right){{\left( {{\gamma _{th}}} \right)}^{\frac{2}{\beta }}}{r^{2 - \beta }} - {\gamma _{th}} \nonumber\\
& {{{\delta ^2}} \mathord{\left/
 {\vphantom {{{\gamma _{th}}{\delta ^2}} {{P_{ap}}}}} \right.
 \kern-\nulldelimiterspace} {{P_{ap}}}} \Big]^{{m_1}} \prod\limits_{l = 2}^m \Big[  - {\lambda _{ap}}{\rho _{ap}}\pi \Gamma \left( {1 + {2 \mathord{\left/
 {\vphantom {2 \beta }} \right.
 \kern-\nulldelimiterspace} \beta }} \right)\Gamma \left( {1 - {2 \mathord{\left/
 {\vphantom {2 \beta }} \right.
 \kern-\nulldelimiterspace} \beta }} \right){{\left( {{\gamma _{th}}} \right)}^{\frac{2}{\beta }}}\nonumber\\
 &\prod\limits_{j = 0}^{l - 1} {\left( {{2 \mathord{\left/
 {\vphantom {2 \beta }} \right.
 \kern-\nulldelimiterspace} \beta } - j} \right){r^{2 - l\beta }}}  \Big]^{{m_l}} dr.
 \end{align}

\end{lemma}

\begin{proof}
See Appendix C.
\end{proof}

\begin{lemma}\label{Lemma4}
The CDF of SINR at the most detrimental eavesdropper which intercepts the  transmission between the typical access point and the sensor is derived as
\begin{align}\label{CDFSK_E}
{F_{{\gamma _{ap,e}}}}\left( x \right) &= \exp \Big\{ { - \pi \lambda _e^{ap}\int_0^\infty  {\exp \big\{ { - {\lambda _{ap}}{\rho _{ap}}\pi \Gamma \left( {1 + {2 \mathord{\left/
 {\vphantom {2 \beta }} \right.
 \kern-\nulldelimiterspace} \beta }} \right)} } }
\nonumber \\ &  { {\Gamma \left( {1 - {2 \mathord{\left/
 {\vphantom {2 \beta }} \right.
 \kern-\nulldelimiterspace} \beta }} \right){\gamma _{th}}^{\frac{2}{\beta }}t - {{{\sigma ^2}{\gamma _{th}}{t^{{\beta  \mathord{\left/
 {\vphantom {\beta  2}} \right.
 \kern-\nulldelimiterspace} 2}}}} \mathord{\left/
 {\vphantom {{{\sigma ^2}{\gamma _{th}}{t^{{\beta  \mathord{\left/
 {\vphantom {\beta  2}} \right.
 \kern-\nulldelimiterspace} 2}}}} {{P_{ap}}}}} \right.
 \kern-\nulldelimiterspace} {{P_{ap}}}}} \big\}dt} \Big\}.
\end{align}

\end{lemma}

\begin{proof}
See Appendix D.
\end{proof}
\subsubsection{Average Secrecy Rate}
The average secrecy rate   between the access point and the sink is the average of the secrecy rate $C_s^{sk}$ over $\gamma_{sk}$ and $\gamma_{ap,e}$, which is given by
\begin{align}\label{second_Hop_ASR}
{{\bar C}_s^{sk}} &= \frac{1}{{\ln 2}}\int_0^\infty  {\frac{{{F_{{\gamma _{sk}}}}\left( x \right)}}{{1 + x}}(1 - {F_{{\gamma _{ap,e}}}}\left( x \right))} dx.
\end{align}
By substituting the CDF of ${\gamma _{sk}}$ in \eqref{CDFSK_Sink_2hop} and the CDF of  ${\gamma _{ap,e}}$ in \eqref{CDFSK_E} into \eqref{second_Hop_ASR},  we can obtain the average secrecy rate between the access point and the sink.

Note that the derived   average secrecy rate between  the access point and the sink is also not in a simple form, we present the interference-limited case for the  average secrecy rate with single antenna at the access point in the following corollary.
\begin{cor}
When the access points are equipped with single antenna in the interference-limited scenario, the  average secrecy rate between the access point and the sink
is given by
\begin{align}\label{second_Hop_ASR_special}
&\bar C_s^{sk} = \frac{{\pi {\lambda _{sk}}}}{{\ln 2}}\int_0^\infty  {\frac{{\exp \left\{ { - {{\pi \lambda _e^{ap}} \mathord{\left/
 {\vphantom {{\pi \lambda _e^{ap}} {{\Lambda _2}{x^{{2 \mathord{\left/
 {\vphantom {2 \beta }} \right.
 \kern-\nulldelimiterspace} \beta }}}}}} \right.
 \kern-\nulldelimiterspace} {{\Lambda _2}{x^{{2 \mathord{\left/
 {\vphantom {2 \beta }} \right.
 \kern-\nulldelimiterspace} \beta }}}}}} \right\}}}{{\left( {1 + x} \right)\left( {{\Lambda _2}{x^{{2 \mathord{\left/
 {\vphantom {2 \beta }} \right.
 \kern-\nulldelimiterspace} \beta }}} + \pi {\lambda _{sk}}} \right)}}dx}  ,
\end{align}
where ${\Lambda _2} = {\lambda _{ap}}{\rho _{ap}}\pi \Gamma \left( {1 + {2 \mathord{\left/
 {\vphantom {2 \beta }} \right.
 \kern-\nulldelimiterspace} \beta }} \right)\Gamma \left( {1 - {2 \mathord{\left/
 {\vphantom {2 \beta }} \right.
 \kern-\nulldelimiterspace} \beta }} \right).$ Based on \eqref{second_Hop_ASR_special}, for a specific target average secrecy rate $\bar C_0$ between the access point and the sink, the number
of sinks must satisfy
\begin{align}\label{lambda_sk}
\lambda _{sk}>{{\bar C}_{\rm{0}}}{\Lambda _{\rm{2}}}\frac{{\ln 2}}{{\pi \varepsilon }},
\end{align}
where $\varepsilon  = \int_{\rm{0}}^\infty  {\frac{{\exp \left\{ { - \pi \lambda _e^{ap}/\left( {{\Lambda _2}{x^{2/\beta }}}
 \right)} \right\}}}{{\left( {1 + x} \right){x^{2/\beta }}}}dx}$.
 \end{cor}

\subsection{Overall Average Secrecy Rate}
In this subsection, we derive the overall average secrecy rate in three-tier WSNs. The instantaneous secrecy rate is defined as $C_s=\min\left({C_s^{ap}},{C_s^{sk}} \right)$. As such, the overall average secrecy rate is calculated as
\begin{align}\label{Overall_ASR}
{\bar C_s} = \int_0^\infty  x{f_{{C_s}}} \left( x \right)dx = \int_0^\infty  {\left( {1 - {F_{{C_s}}}\left( x \right)} \right)} dx,
\end{align}
where ${f_{{C_s}}} \left( x \right)$ and ${F_{{C_s}}}\left( x \right)$ is the probability density function (PDF) and the CDF of $C_s$, respectively. The CDF of $C_s$ is calculated as
\begin{align} \label{CDF_C_s_eq}
{F_{{C_s}}}\left( x \right)&=\Pr\left(\min\left({C_s^{ap}},{C_s^{sk}} \right)<x\right)\nonumber\\
&=1-\Pr\left(\min\left({C_s^{ap}},{C_s^{sk}} \right)>x\right)\nonumber\\
&=1-\Pr\left(C_s^{ap}>x\right)\Pr\left({C_s^{sk}}>x\right).
\end{align}
Substituting \eqref{CDF_C_s_eq} into \eqref{Overall_ASR}, we have
\begin{align}\label{Overall_ASR_12}
{\bar C_s} =  \int_0^\infty  { \Pr\left(C_s^{ap}>x\right)\Pr\left({C_s^{sk}}>x\right)} dx,
\end{align}
where
\begin{align}\label{CDF_C_s_ap_1}
\Pr\left(C_s^{ap}>x\right)=1-\int_0^\infty  {{f_{{\gamma _{s,e}}}}\left( t \right){F_{{\gamma _{ap}}}}\left( {{2^x}\left( {1 + t} \right) - 1} \right)} dt
\end{align}
and
\begin{align}\label{CDF_C_s_sk_1}
\Pr\left({C_s^{sk}}>x\right)=1-\int_0^\infty  {{f_{{\gamma _{ap,e}}}}\left( t \right){F_{{\gamma _{sk}}}}\left( {{2^x}\left( {1 + t} \right) - 1} \right)} dt.
\end{align}
 Here, ${f_{{\gamma _{s,e}}}}$ is the derivative of ${F_{{\gamma _{s,e}}}}$ given in \eqref{CDFAP_E}, and ${f_{{\gamma _{ap,e}}}}$ is the derivative of ${F_{{\gamma _{ap,e}}}}$ given in \eqref{CDFSK_E}.

\begin{figure}[t!]
    \begin{center}
        \includegraphics[width=2.8 in]{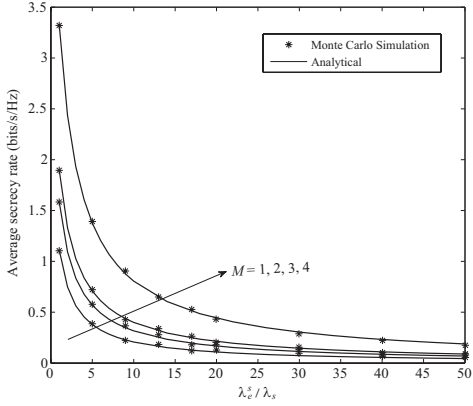}
        \caption{The average secrecy rate versus $\frac{{\lambda _e^s}}{{{\lambda _s}}}$.  $\lambda_{s}=10^{-2}$, $\rho_s=0.01$, $\lambda_{ap}=10^{-2}$, $\rho_{ap}=0.1$, $\alpha=3.5$, $P_{ap}=25$ dBm.}
        \label{Fig2}
    \end{center}
\end{figure}

Unfortunately, the derived overall average secrecy rate between the sensor and the sink is not in a simple form, which motivates us to consider the interference-limited case with single antenna at the access point, as presented in the following corollary.
\begin{cor}
When the access points are equipped with single antenna in the interference-limited scenario, the overall average secrecy rate between the sensor and the sink
is given by
\begin{align}\label{overal_ASR_special}
&\bar C_s =  \int_0^\infty  {\Bigg[ {\int_0^\infty  {\frac{{2\pi \lambda _e^s}}{{\alpha {\Lambda _1}{y^{{2 \mathord{\left/
 {\vphantom {2 \alpha }} \right.
 \kern-\nulldelimiterspace} \alpha } + 1}}}}\exp \big\{ { - {{\pi \lambda _e^s} \mathord{\left/
 {\vphantom {{\pi \lambda _e^s} {\left( {{\Lambda _1}{y^{{2 \mathord{\left/
 {\vphantom {2 \alpha }} \right.
 \kern-\nulldelimiterspace} \alpha }}}} \right)}}} \right.
 \kern-\nulldelimiterspace} {\big( {{\Lambda _1}{y^{{2 \mathord{\left/
 {\vphantom {2 \alpha }} \right.
 \kern-\nulldelimiterspace} \alpha }}}} \big)}}} \big\}} } \Bigg.}
 \nonumber \\
&\Bigg. {\frac{{\pi {\lambda _{ap}}\left( {1 - {\rho _{ap}}} \right)}}{{{\Lambda _1}{{\left( {{2^x}\left( {1 + y} \right) - 1} \right)}^{{2 \mathord{\left/
 {\vphantom {2 \alpha }} \right.
 \kern-\nulldelimiterspace} \alpha }}} + \pi {\lambda _{ap}}\left( {1 - {\rho _{ap}}} \right)}}dy} \Bigg]
  \nonumber \\
&\Bigg[ {\int_0^\infty  {\frac{{2{\pi ^2}\lambda _e^{ap}{\lambda _{sk}}\exp \left\{ { - {{\pi \lambda _e^{ap}} \mathord{\left/
 {\vphantom {{\pi \lambda _e^{ap}} {{\Lambda _2}{y^{{2 \mathord{\left/
 {\vphantom {2 \beta }} \right.
 \kern-\nulldelimiterspace} \beta }}}}}} \right.
 \kern-\nulldelimiterspace} {{\Lambda _2}{y^{{2 \mathord{\left/
 {\vphantom {2 \beta }} \right.
 \kern-\nulldelimiterspace} \beta }}}}}} \right\}}}{{\beta {\Lambda _2}{y^{{2 \mathord{\left/
 {\vphantom {2 \beta }} \right.
 \kern-\nulldelimiterspace} \beta } + 1}}\left( {{\Lambda _2}{{\left( {{2^x}\left( {1 + y} \right) - 1} \right)}^{{2 \mathord{\left/
 {\vphantom {2 \beta }} \right.
 \kern-\nulldelimiterspace} \beta }}} + \pi {\lambda _{sk}}} \right)}}dy} } \Bigg]dx ,
\end{align}
where ${\Lambda _1} = \left( {{\lambda _s}{\rho _s} + {\lambda _{ap}}{\rho _{ap}}{\mu ^{\frac{2}{\alpha }}}} \right)\pi \Gamma \left( {1 + {2 \mathord{\left/
 {\vphantom {2 \alpha }} \right.
 \kern-\nulldelimiterspace} \alpha }} \right)\Gamma \left( {1 - {2 \mathord{\left/
 {\vphantom {2 \alpha }} \right.
 \kern-\nulldelimiterspace} \alpha }} \right)$ and ${\Lambda _2} = {\lambda _{ap}}{\rho _{ap}}\pi \Gamma \left( {1 + {2 \mathord{\left/
 {\vphantom {2 \beta }} \right.
 \kern-\nulldelimiterspace} \beta }} \right)\Gamma \left( {1 - {2 \mathord{\left/
 {\vphantom {2 \beta }} \right.
 \kern-\nulldelimiterspace} \beta }} \right).$
 \end{cor}

\section{Numerical Examples}
In this section, we present numerical examples to show the average secrecy rate of the three-tier WSN. We assume that the sensor's transmit power $P_s=15$ dBm, the power spectral density of noise is $-170$ dBm/Hz, and the bandwidth is 1 MHz. {We also assume  that all the channel gains follow a complex Gaussian distribution with zero
mean and unit variance.} In all the figures, we see a precise match between the simulations and the exact analytical curves, which validate our analysis.


\subsection{Average Secrecy Rate between  Sensor and Access Point}


Fig.~\ref{Fig2} plots the average secrecy rate between the sensor and the access point  versus ${{\lambda _e^s} \mathord{\left/
 {\vphantom {{\lambda _e^s} {{\lambda _s}}}} \right.
 \kern-\nulldelimiterspace} {{\lambda _s}}}$. The analytical results are obtained from \eqref{first_Hop_ASR}. We first see that the average secrey rate decreases with increasing the density of eavesdroppers that intercepts the transmission between sensor and access point, due to the detrimental effects of eavesdropping.  We also see that the average secrecy rate increases with increasing the number of antennas at the access point, which results from the array again brought by using MRC at the access point.

\begin{figure}[t!]
    \begin{center}
        \includegraphics[width=2.8 in]{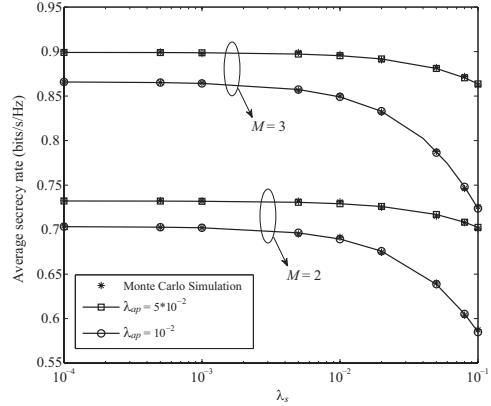}
        \caption{The average secrecy rate versus ${{\lambda _s}}$. $\rho_s=0.05$, $\rho_{ap}=0.5$, $\lambda _e^s=10^{-3}$, $\alpha=3.5$, $P_{ap}=25$ dBm.}
        \label{Fig3}
    \end{center}
\end{figure}

Fig.~\ref{Fig3} plots the average secrecy rate between the sensor and the access point  versus  $\lambda_s$ for various $\lambda_{ap}$ and $M$. The analytical results are obtained from \eqref{first_Hop_ASR}. An interesting observation is that for the same number of antennas $M$, the average secrecy rate is nearly invariable for $\lambda_s<2  \times
{10^{ - 3}}$, since the interference from  other sensors is much smaller than the interference from the active access points, and slightly increasing the interference from the sensor imposes negligible effect on the performance.  However, when  $\lambda_s>2  \times
{10^{ - 3}}$, the interference from other sensors is comparable with the interference from the active access points, and increasing the interference from the sensor degrades the secrecy performance. We also observe that increasing $\lambda_{ap}$ increases the average secrecy rate. This is because with more access points, the distance between the typical sensor and the typical access point becomes shorter, which improves the average secrecy rate. In addition, we find that increasing $\lambda_{ap}$  slows down the decreasing trend of average secrecy rate when $\lambda_s$ increases.

\subsection{Average Secrecy Rate between  Access point and Sink}

\begin{figure}[t!]
    \begin{center}
        \includegraphics[width=2.8 in]{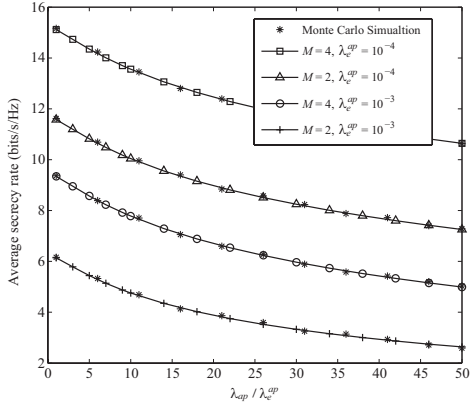}
        \caption{The average secrecy rate versus $\frac{{{\lambda _{ap}}}}{{\lambda _e^{ap}}}$. $\rho_{ap}=0.1$, $\lambda_{sk}=10^{-2}$, $\beta=3.5$, $P_{ap}=15$ dBm.}
        \label{FIg4}
    \end{center}
\end{figure}

\begin{figure}[t!]
    \begin{center}
    \vspace{-0.6cm}
        \includegraphics[width=2.8 in]{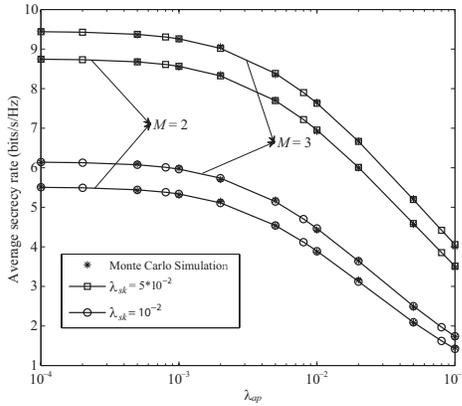}
        \caption{The average secrecy rate versus ${{\lambda _{ap}}}$. $\rho_{ap}=0.1$, $\beta=3$, $\lambda _e^{ap}=10^{-3}$, $P_{ap}=25$ dBm,}
        \label{Fig5}
    \end{center}
\end{figure}


Fig.~\ref{FIg4} plots the average secrecy rate between the access point and the sink  versus   ${{\lambda _e^{ap}} \mathord{\left/
 {\vphantom {{\lambda _e^{ap}} {{\lambda _{ap}}}}} \right.
 \kern-\nulldelimiterspace} {{\lambda _{ap}}}}$ for various ${\lambda _{ap}}$ and $M$. The analytical results are obtained from \eqref{second_Hop_ASR}.
 We first observe that  the average secrecy rate decreases with increasing ${{\lambda _e^{ap}} \mathord{\left/
 {\vphantom {{\lambda _e^{ap}} {{\lambda _{ap}}}}} \right.
 \kern-\nulldelimiterspace} {{\lambda _{ap}}}}$, which indicates that more access points need to be deployed as the density of eavesdroppers increases, to combat eavesdropping. Second, with the same number of antennas at the access point, the average secrecy rate decreases with increasing ${\lambda _e^{ap}}$. The average secrecy rate between  the access point and the sink improves with increasing the number of antennas at the access point $M$.

Fig.~\ref{Fig5}  plots the average secrecy rate between the access point and the sink  versus $\lambda_{ap}$ for various ${\lambda _{sk}}$ and $M$. The analytical results are obtained from \eqref{second_Hop_ASR}.
We observe that the average secrecy rate alters slightly for $\lambda_{ap}<2  \times
{10^{ - 3}}$, and decreases with increasing $\lambda_{ap}$ for $\lambda_{ap}>2  \times
{10^{ - 3}}$. This can be explained by the fact that  for $\lambda_{ap}<2  \times
{10^{ - 3}}$, the interference from the active access points is  relatively small compared with the noise, and increasing the number of access points scarcely influence the performance. However, for $\lambda_{ap}>2  \times
{10^{ - 3}}$, the interference from the access point imposes a dominant impact on the SINR between the access point and the sink, thus increasing the interference from the access points degrades the average secrecy rate. Another observation is that the average secrecy rate improves with increasing the  density of sink, because the distance between the typical access point and the corresponding sink becomes shorter.

\subsection{Overall Average Secrecy Rate}

\begin{figure}[t!]
    \begin{center}
        \includegraphics[width=2.8 in]{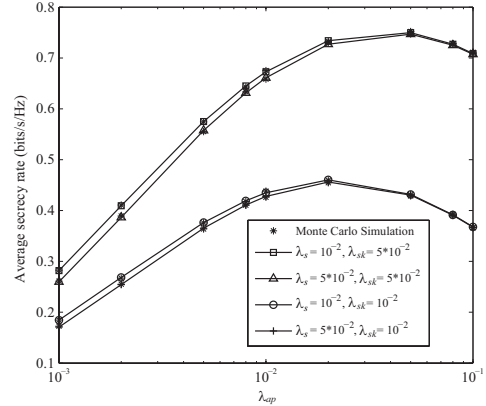}
        \caption{The average secrecy rate versus ${{\lambda _{ap}}}$. $P_{ap}=30$ dBm, $M=2$, $\rho_s=0.01$, $\rho_{ap}=0.1$, $\alpha=2.8$, $\beta=3.2$,  $\lambda _e^s=\lambda _e^{ap}=5*10^{-3}$.}
        \label{Fig6}
    \end{center}
\end{figure}

\begin{figure}[t!]
    \begin{center}
        \includegraphics[width=2.8 in]{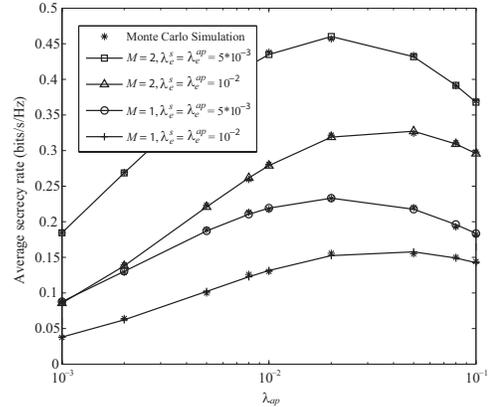}
        \caption{The average secrecy rate versus ${{\lambda _{ap}}}$. $P_{ap}=30$ dBm, $\rho_s=0.01$, $\rho_{ap}=0.1$, $\alpha=2.8$, $\beta=3.2$, $\lambda_s=\lambda_{sk}=10^{-2}$.}
        \label{Fig7}
    \end{center}
\end{figure}


Fig.~\ref{Fig6} plots the overall average secrecy rate versus $\lambda_{ap}$ for various $\lambda_{s}$ and $\lambda_{sk}$. The analytical results are obtained from \eqref{Overall_ASR_12}. Interestingly, we find that the overall average secrecy rate first increases, and then decreases with  increasing $\lambda_{ap}$, which implies that there exists an optimal $\lambda_{ap}$ to achieve the maximum average secrecy rate. This phenomenon can be well explained by the tradeoff between the benefits brought by  the shorter distance from the typical sensor to the typical access point   and the detrimental effects caused by  more interference from the active access points due to increasing $\lambda_{ap}$. It is also seen that the overall average secrecy rate can be improved by deploying more sinks, due to the shorter distance between the access point and the sink. It is further  demonstrated that deploying more sensors in this network may not greatly degrade the average secrecy rate due to the low transmit power of sensors. More importantly, it is shown that the optimal $\lambda_{ap}$ is more dependent on the $\lambda_{sk}$.

Fig.~\ref{Fig7} plots the overall average secrecy rate versus $\lambda_{ap}$ for various $\lambda_e^{s}$, $\lambda_e^{ap}$ and $M$. The analytical results are obtained from \eqref{Overall_ASR_12}. Similar as Fig.~\ref{Fig6}, we see that the overall average secrecy rate first increases, and then decreases with  increasing $\lambda_{ap}$. As expected, the average secrecy rate decreases with increasing eavesdroppers. It is indicated that the optimal $\lambda_{ap}$ for achieving the maximum average secrecy rate does not alter drastically with different $\lambda_e^{s}$ and $\lambda_e^{ap}$.

\section{Conclusion}
 {We have analyzed the physical layer security of three-tier WSNs. We have examined the impact of random locations and spatial densities of sensors, access points, sinks, and external eavesdroppers on the secrecy performance. We have also obtained new expressions for the average secrecy rate. Based on our analysis, we have established the importance of physical layer security in three-tier WSNs, where our results support useful guidelines on secure transmission in practical WSNs. An important result is the minimum number of sinks required for a target average secrecy rate, which facilitates secure node deployment design in WSNs.}



\appendices
\numberwithin{equation}{section}

\section{  Proof of Lemma 1}\label{aplemma1}
From \eqref{SINR_AP}, the CDF of ${\gamma _{ap}}$ is given by
\begin{align}\label{CDFAP_D1}
{F_{{\gamma _{ap}}}}\left( {{\gamma _{th}}} \right) =& \int_0^\infty  {\Pr \Big[ {\frac{{{{\left\| {{{\bf{h}}_{{s_0},a{p_0}}}} \right\|}^2}{r^{ - \alpha }}}}{{I{n_{ap}} + {{{\delta ^2}} \mathord{\left/
 {\vphantom {{{\delta ^2}} {{P_s}}}} \right.
 \kern-\nulldelimiterspace} {{P_s}}}}} \le {{\gamma _{th}}}} \Big]} {f_{\left| {{X_{{s_0},a{p_0}}}} \right|}}\left( r \right)dr \nonumber \\
  = & \int_0^\infty  {\Pr \Big[ {\frac{{{{\left\| {{{\bf{h}}_{{s_0},a{p_0}}}} \right\|}^2}{{\left| {{X_{{s_0},a{p_0}}}} \right|}^{ - \alpha }}}}{{I{n_{ap}} + {{{\delta ^2}} \mathord{\left/
 {\vphantom {{{\delta ^2}} {{P_s}}}} \right.
 \kern-\nulldelimiterspace} {{P_s}}}}} \le {\gamma _{th}}} \Big]} 2\pi {\lambda _{ap}}
 \nonumber \\
 & \left( {1 - {\rho _{ap}}} \right)r\exp \left( { - \pi {\lambda _{ap}}\left( {1 - {\rho _{ap}}} \right){r^2}} \right)dr,
\end{align}
where ${f_{\left| {{X_{{s_0},a{p_0}}}} \right|}}\left( r \right)$ is the PDF of the nearest distance between the access point and the typical sensor.
The CDF of the access point SINR at distance $r$ from its corresponding sensor  is given as
\begin{align}\label{CDFAP_D2}
\Pr \Big[ {\frac{{{{\left\| {{{\bf{h}}_{{s_0},a{p_0}}}} \right\|}^2}{r^{ - \alpha }}}}{{I{n_{ap}} + {{{\delta ^2}} \mathord{\left/
 {\vphantom {{{\delta ^2}} {{P_s}}}} \right.
 \kern-\nulldelimiterspace} {{P_s}}}}} \le {\gamma _{th}}} \Big] &
  \nonumber \\
 & \hspace{-4.3cm}   = 1 - \sum\limits_{m = 0}^{M - 1} {\frac{1}{{m!}}{\mathbbm{E}_{{\Phi _{s,a}}}}\Big\{ {{\mathbbm{E}_{{\Phi _{ap,a}}}}\Big\{ {\int_0^\infty  {{{\left[ {{\gamma _{th}}{r^\alpha }\left( {\tau  + {{{\delta ^2}} \mathord{\left/
 {\vphantom {{{\delta ^2}} {{P_s}}}} \right.
 \kern-\nulldelimiterspace} {{P_s}}}} \right)} \right]}^m}} } } }
  \nonumber \\
 & \hspace{-3.6cm} { {\exp \left[ { - {\gamma _{th}}{r^\alpha }\left( {\tau  + {{{\delta ^2}} \mathord{\left/
 {\vphantom {{{\delta ^2}} {{P_s}}}} \right.
 \kern-\nulldelimiterspace} {{P_s}}}} \right)} \right]d\Pr \left( {I{n_{ap}} \le \tau } \right)} \Big\}} \Big\}.
\end{align}

We then substitute $ {\left( { - \left( {\tau  + {{{\delta ^2}} \mathord{\left/
 {\vphantom {{{\delta ^2}} {{P_s}}}} \right.
 \kern-\nulldelimiterspace} {{P_s}}}} \right){\gamma _{th}}} \right)^m}{e^{ - \left( {\tau  + {{{\delta ^2}} \mathord{\left/
 {\vphantom {{{\delta ^2}} {{P_s}}}} \right.
 \kern-\nulldelimiterspace} {{P_s}}}} \right)\gamma _{th}^{\left\{ s \right\}}{r^\alpha }}}$\\$={\left. {\frac{{{d^m}\left( {{e^{ - {\gamma _{th}}x\left( {\tau  + {{{\delta ^2}} \mathord{\left/
 {\vphantom {{{\delta ^2}} {{P_s}}}} \right.
 \kern-\nulldelimiterspace} {{P_s}}}} \right)}}} \right)}}{{d{x^m}}}} \right|_{x = {r^\alpha }}}$ into \eqref{CDFAP_D2}, we rewrite the CDF of the access point SINR at distance $r$ from its corresponding sensor as
 \begin{align}
\Pr \Big[ {\frac{{{{\left\| {{{\bf{h}}_{{s_0},a{p_0}}}} \right\|}^2}{r^{ - \alpha }}}}{{I{n_{ap}} + {{{\delta ^2}} \mathord{\left/
 {\vphantom {{{\delta ^2}} {{P_s}}}} \right.
 \kern-\nulldelimiterspace} {{P_s}}}}} \le {\gamma _{th}}} \Big] =& 1 - {\mathbbm{E}_{{\Phi _{s,a}}}}\Big\{ {{\mathbbm{E}_{{\Phi _{ap,a}}}}\Big\{ {} }
 \nonumber \\
 & \hspace{-3.6cm}  { {\int_0^\infty  {\exp \left[ { - {\gamma _{th}}{r^\alpha }\left( {\tau  + {{{\delta ^2}} \mathord{\left/
 {\vphantom {{{\delta ^2}} {{P_s}}}} \right.
 \kern-\nulldelimiterspace} {{P_s}}}} \right)} \right]d\Pr \left( {I{n_{ap}} \le \tau } \right)} } \Big\}} \Big\}
  \nonumber \\
 & \hspace{-3.6cm} - \sum\limits_{m = 1}^{M - 1} {\frac{{{{\left( {{r^\alpha }} \right)}^m}}}{{m!{{\left( { - 1} \right)}^m}}}{\mathbbm{E}_{{\Phi _{s,a}}}}\Bigg\{ {{\mathbbm{E}_{{\Phi _{ap,a}}}}\Bigg\{ {} } }
  \nonumber \\
 & \hspace{-3.6cm}{\int_0^\infty  {{{\left. {\frac{{{d^m}\left( {{e^{ - {\gamma _{th}}x\left( {\tau  + {{{\delta ^2}} \mathord{\left/
 {\vphantom {{{\delta ^2}} {{P_s}}}} \right.
 \kern-\nulldelimiterspace} {{P_s}}}} \right)}}} \right)}}{{d{x^m}}}} \right|}_{x = {r^\alpha }}}} d\Pr \left( {I{n_{ap}} \le \tau } \right)} \Bigg\} \Bigg\}
  \nonumber 
\end{align}
\begin{align} \label{CDFAP_D3}
 &  = 1 - \exp \left( { - {{{\gamma _{th}}{r^\alpha }{\delta ^2}} \mathord{\left/
 {\vphantom {{{\gamma _{th}}{r^\alpha }{\delta ^2}} {{P_s}}}} \right.
 \kern-\nulldelimiterspace} {{P_s}}}} \right){\mathcal{L}_{I{n_{ap}}}}\left( {{\gamma _{th}}{r^\alpha }} \right)
 \nonumber \\ 
 &  - \sum\limits_{m = 1}^{M - 1} {\frac{{{{\left( {{r^\alpha }} \right)}^m}}}{{m!{{\left( { - 1} \right)}^m}}}{{\left. {\frac{{{d^m}\left( {\exp \left( { - {{{\gamma _{th}}x{\delta ^2}} \mathord{\left/
 {\vphantom {{{\gamma _{th}}x{\delta ^2}} {{P_s}}}} \right.
 \kern-\nulldelimiterspace} {{P_s}}}} \right){\mathcal{L}_{I{n_{ap}}}}\left( {{\gamma _{th}}x} \right)} \right)}}{{d{x^m}}}} \right|}_{x = {r^\alpha }}}} .
\end{align}
Remind that ${I_{s,ap}} = {\sum _{i \in {\Phi _{s,a}}\backslash \left\{ {{s_0}} \right\}}}{\left| {\frac{{{{\bf{h}}_{{s_0},a{p_0}}}^\dag }}{{\left\| {{{\bf{h}}_{{s_0},a{p_0}}}} \right\|}}{{\bf{h}}_{i,a{p_0}}}} \right|^2}{\left| {{X_{i,a{p_0}}}} \right|^{ - \alpha }}$, using Slivnyak's theorem, the Laplace transform of $I_{s,ap}$ is
  \begin{align}\label{LAP_sap}
&{\mathcal{L}_{{I_{s,ap}}}}\left( s \right)  \nonumber \\
 &\hspace{0.1cm} ={\mathbbm{E}_{{\Phi _s}}}\left[ {\exp \left\{ { - s{\sum _{i \in {\Phi _{s,a}}\backslash \left\{ {{s_0}} \right\}}}{{\left| {\frac{{{{\bf{h}}_{{s_0},a{p_0}}}^\dag }}{{\left\| {{{\bf{h}}_{{s_0},a{p_0}}}} \right\|}}{{\bf{h}}_{i,a{p_0}}}} \right|}^2}{{\left| {{X_{i,a{p_0}}}} \right|}^{ - \alpha }}} \right\}} \right]
 \nonumber \\
 &\hspace{0.1cm} \mathop  = \limits^{\left( a \right)}
\exp \left\{ { - 2\pi {\lambda _s}{\rho _s}\int_0^\infty  {\left( {1 - {\mathcal{L}_{\frac{{{{\bf{h}}_{{s_0},a{p_0}}}^\dag }}{{\left\| {{{\bf{h}}_{{s_0},a{p_0}}}} \right\|}}{{\bf{h}}_{i,a{p_0}}}}}\left( {s{y^{ - \alpha }}} \right)} \right)ydy} } \right\}
  \nonumber \\
 & \hspace{0.1cm} \mathop  = \limits^{\left( b \right)} \exp \left\{ { - 2\pi {\lambda _s}{\rho _s}\int_0^\infty  {\left( {1 - \frac{1}{{1 + s{y^{ - \alpha }}}}} \right)ydy} } \right\}
  \nonumber \\
 & \hspace{0.1cm}=\exp \left\{ { - {\lambda _s}{\rho _s}\pi \Gamma \left( {1 + {2 \mathord{\left/
 {\vphantom {2 \alpha }} \right.
 \kern-\nulldelimiterspace} \alpha }} \right)\Gamma \left( {1 - {2 \mathord{\left/
 {\vphantom {2 \alpha }} \right.
 \kern-\nulldelimiterspace} \alpha }} \right){s^{{2 \mathord{\left/
 {\vphantom {2 \alpha }} \right.
 \kern-\nulldelimiterspace} \alpha }}}} \right\},
\end{align}
In \eqref{LAP_sap}, $(a)$ follows from  the generating functionnal of HPPP in \cite{stoyanstochastic}, $(b)$ follows from  the fact that  ${\left| {\frac{{{{\bf{h}}_{{s_0},a{p_0}}}^\dag }}{{\left\| {{{\bf{h}}_{{s_0},a{p_0}}}} \right\|}}{{\bf{h}}_{i,a{p_0}}}} \right|^2} \sim \exp \left( 1 \right)$.

Since ${I_{ap,ap}} = \mu {\sum _{j \in {\Phi _{ap,a}}\backslash \left\{ {a{p_0}} \right\}}}{\left| {\frac{{{{\bf{h}}_{{s_0},a{p_0}}}^\dag }}{{\left\| {{{\bf{h}}_{{s_0},a{p_0}}}} \right\|}}{{\bf{H}}_{j,a{p_0}}}\frac{{{{\bf{h}}_{j,s{k_j}}}^\dag }}{{\left\| {{{\bf{h}}_{j,s{k_j}}}} \right\|}}} \right|^2}$
\\
${\left| {{X_{j,a{p_0}}}} \right|^{ - \alpha }} = \mu {\sum _{j \in {\Phi _{ap}}\backslash \left\{ {a{p_0}} \right\}}}{H_j^{ap,ap}}{\left| {{X_{j,a{p_0}}}} \right|^{ - \alpha }}$, the Laplace transform of ${I_{ap,ap}}$ is
\begin{align}\label{LAP_apap}
&{\mathcal{L}_{{I_{ap,ap}}}}\left( s \right)  \nonumber \\
 &\hspace{0.3cm} =\exp \left( { - \smallint \left[ {1 - {\mathbbm{E}_{h}}\left( {\exp \left( { - s\mu {H_j^{ap,ap}}{y^{ - \alpha }}} \right)} \right)} \right]{\lambda _{ap}}{\rho _{ap}}2\pi ydy} \right)
 \nonumber \\
 &\hspace{0.3cm} \mathop  = \limits^{\left( c\right)}\exp \left\{ { - {\lambda _{ap}}{\rho _{ap}}\pi {\mu ^{\frac{2}{\alpha }}}{\mathbbm{E}_h}\left\{ {\left({H_j^{ap,ap}}\right)^{\frac{2}{\alpha }}} \right\}\Gamma \left( {1 - \frac{2}{\alpha }} \right){s^{\frac{2}{\alpha }}}} \right\}
  \nonumber \\
 & \hspace{0.3cm} \mathop  = \limits^{\left( d\right)} \exp \left\{ { - {\lambda _{ap}}{\rho _{ap}}\pi {\mu ^{\frac{2}{\alpha }}}\Gamma \left( {1 + {2 \mathord{\left/
 {\vphantom {2 \alpha }} \right.
 \kern-\nulldelimiterspace} \alpha }} \right)\Gamma \left( {1 - {2 \mathord{\left/
 {\vphantom {2 \alpha }} \right.
 \kern-\nulldelimiterspace} \alpha }} \right){s^{{2 \mathord{\left/
 {\vphantom {2 \alpha }} \right.
 \kern-\nulldelimiterspace} \alpha }}}} \right\},
\end{align}
where $(c)$ follows from the generating functionnal of HPPP in \cite{stoyanstochastic},
 $(d)$ follows from $H_j\sim \exp \left( 1 \right)$.

 With the Laplace transform of ${I_{s,ap}}$ and ${I_{ap,ap}}$, we derive the Laplace transform of $In_{ap}$ as
\begin{align}\label{LAP_inap_sensor}
&{\mathcal{L}_{I{n_{ap}}}}\left( s \right)= {\mathcal{L}_{{I_{s,ap}}}}\left( s \right) {\mathcal{L}_{{I_{ap,ap}}}}\left( s \right) =  \nonumber\\
&\exp \left\{ { - \left( {{\lambda _s}{\rho _s} + {\lambda _{ap}}{\rho _{ap}}{\mu ^{\frac{2}{\alpha }}}} \right)\pi \Gamma \left( {1 + {2 \mathord{\left/
 {\vphantom {2 \alpha }} \right.
 \kern-\nulldelimiterspace} \alpha }} \right)\Gamma \left( {1 - {2 \mathord{\left/
 {\vphantom {2 \alpha }} \right.
 \kern-\nulldelimiterspace} \alpha }} \right){s^{{2 \mathord{\left/
 {\vphantom {2 \alpha }} \right.
 \kern-\nulldelimiterspace} \alpha }}}} \right\}.
\end{align}

Substituting \eqref{LAP_inap_sensor} into \eqref{CDFAP_D3}, we obtain
\begin{align}\label{CDFAP_D5}
&\hspace{-0.5cm}\Pr \Big[ {\frac{{{{\left\| {{{\bf{h}}_{{s_0},a{p_0}}}} \right\|}^2}{r^{ - \alpha }}}}{{I{n_{ap}} + {{{\delta ^2}} \mathord{\left/
 {\vphantom {{{\delta ^2}} {{P_s}}}} \right.
 \kern-\nulldelimiterspace} {{P_s}}}}} \le {\gamma _{th}}} \Big]  = 1 - \exp\left\{ { - \left( {{\lambda _s}{\rho _s} + {\lambda _{ap}}{\rho _{ap}}{\mu ^{\frac{2}{\alpha }}}} \right)}\right.\nonumber\\
&\left.{\pi \Gamma \left( {1 + {2 \mathord{\left/
 {\vphantom {2 \alpha }} \right.
 \kern-\nulldelimiterspace} \alpha }} \right)\Gamma \left( {1 - {2 \mathord{\left/
 {\vphantom {2 \alpha }} \right.
 \kern-\nulldelimiterspace} \alpha }} \right){{\left( {{\gamma _{th}}} \right)}^{{2 \mathord{\left/
 {\vphantom {2 \alpha }} \right.
 \kern-\nulldelimiterspace} \alpha }}}{r^2}}{ - {{{\gamma _{th}}{r^\alpha }{\delta ^2}} \mathord{\left/
 {\vphantom {{{\gamma _{th}}{r^\alpha }{\delta ^2}} {{P_s}}}} \right.
 \kern-\nulldelimiterspace} {{P_s}}}}\right\}\nonumber\\
 &-\sum\limits_{m = 1}^{M - 1} {\frac{{{{\left( {{r^\alpha }} \right)}^m}}}{{m!{{\left( { - 1} \right)}^m}}}{{\left. {\frac{{{d^m}\left( {V\left( x \right)} \right)}}{{d{x^m}}}} \right|}_{x = {r^\alpha }}}},
\end{align}
where $V\left( x \right) =  \exp \left\{  - \left( {{\lambda _s}{\rho _s} + {\lambda _{ap}}{\rho _{ap}}{\mu ^{\frac{2}{\alpha }}}} \right)\pi \Gamma \left( {1 + {2 \mathord{\left/
 {\vphantom {2 \alpha }} \right.
 \kern-\nulldelimiterspace} \alpha }} \right)\right.$\\$\left.\Gamma \left( {1 - {2 \mathord{\left/
 {\vphantom {2 \alpha }} \right.
 \kern-\nulldelimiterspace} \alpha }} \right){{\left( {{\gamma _{th}}x} \right)}^{{2 \mathord{\left/
 {\vphantom {2 \alpha }} \right.
 \kern-\nulldelimiterspace} \alpha }}} - {{{\gamma _{th}}x{\delta ^2}} \mathord{\left/
 {\vphantom {{{\gamma _{th}}x{\delta ^2}} {{P_s}}}} \right.
 \kern-\nulldelimiterspace} {{P_s}}}\right\}$.

 We then apply the Fa$\grave{a}$
 di Bruno's formula to solve the derivative of $m$th order as follows:
\begin{align}\label{CDFAP_D6}
&\hspace{-0.2cm}\Pr \Big[ {\frac{{{{\left\| {{{\bf{h}}_{{s_0},a{p_0}}}} \right\|}^2}{r^{ - \alpha }}}}{{I{n_{ap}} + {{{\delta ^2}} \mathord{\left/
 {\vphantom {{{\delta ^2}} {{P_s}}}} \right.
 \kern-\nulldelimiterspace} {{P_s}}}}} \le {\gamma _{th}}} \Big] =1 - \exp\Big\{{ - \left( {{\lambda _s}{\rho _s} + {\lambda _{ap}}{\rho _{ap}}{\mu ^{\frac{2}{\alpha }}}} \right)}\nonumber\\
 &{\pi \Gamma \left( {1 + {2 \mathord{\left/
 {\vphantom {2 \alpha }} \right.
 \kern-\nulldelimiterspace} \alpha }} \right)\Gamma \left( {1 - {2 \mathord{\left/
 {\vphantom {2 \alpha }} \right.
 \kern-\nulldelimiterspace} \alpha }} \right)}{{{\left( {{\gamma _{th}}} \right)}^{{2 \mathord{\left/
 {\vphantom {2 \alpha }} \right.
 \kern-\nulldelimiterspace} \alpha }}}{r^2} - {{{\gamma _{th}}{r^\alpha }{\delta ^2}} \mathord{\left/
 {\vphantom {{{\gamma _{th}}{r^\alpha }{\delta ^2}} {{P_s}}}} \right.
 \kern-\nulldelimiterspace} {{P_s}}}}\Big\}-\nonumber\\
 &\sum\limits_{m = 1}^{M - 1} {\frac{{{{\left( {{r^\alpha }} \right)}^m}}}{{{{\left( { - 1} \right)}^m}}}\sum {\frac{1}{{\prod\limits_{l = 1}^m {{m_l}!l{!^{{m_l}}}} }}} }\exp\Big\{{ - \left( {{\lambda _s}{\rho _s} + {\lambda _{ap}}{\rho _{ap}}{\mu ^{{2 \mathord{\left/
 {\vphantom {2 \alpha }} \right.
 \kern-\nulldelimiterspace} \alpha }}}} \right)}\nonumber\\
 &{\pi \Gamma \left( {1 + {2 \mathord{\left/
 {\vphantom {2 \alpha }} \right.
 \kern-\nulldelimiterspace} \alpha }} \right)\Gamma \left( {1 - {2 \mathord{\left/
 {\vphantom {2 \alpha }} \right.
 \kern-\nulldelimiterspace} \alpha }} \right)}{{{\left( {{\gamma _{th}}} \right)}^{{2 \mathord{\left/
 {\vphantom {2 \alpha }} \right.
 \kern-\nulldelimiterspace} \alpha }}}{r^2} - {{{\gamma _{th}}{r^\alpha }{\delta ^2}} \mathord{\left/
 {\vphantom {{{\gamma _{th}}{r^\alpha }{\delta ^2}} {{P_s}}}} \right.
 \kern-\nulldelimiterspace} {{P_s}}}}\Big\}\nonumber\\
 &\Big[ { - {2 \mathord{\left/
 {\vphantom {2 \alpha }} \right.
 \kern-\nulldelimiterspace} \alpha }\left( {{\lambda _s}{\rho _s} + {\lambda _{ap}}{\rho _{ap}}{\mu ^{{2 \mathord{\left/
 {\vphantom {2 \alpha }} \right.
 \kern-\nulldelimiterspace} \alpha }}}} \right)} {\pi \Gamma \left( {1 + {2 \mathord{\left/
 {\vphantom {2 \alpha }} \right.
 \kern-\nulldelimiterspace} \alpha }} \right)\Gamma \left( {1 - {2 \mathord{\left/
 {\vphantom {2 \alpha }} \right.
 \kern-\nulldelimiterspace} \alpha }} \right)}\nonumber\\
 & {{{\left( {{\gamma _{th}}} \right)}^{\frac{2}{\alpha }}}{r^{\left( {2 - \alpha } \right)}}}{ - {{{\gamma _{th}}{\delta ^2}} \mathord{\left/
 {\vphantom {{{\gamma _{th}}{\delta ^2}} {{P_s}}}} \right.
 \kern-\nulldelimiterspace} {{P_s}}}}\Big]^{m_1}\prod\limits_{l = 2}^m \Big[ { - \left( {{\lambda _s}{\rho _s} + {\lambda _{ap}}{\rho _{ap}}{\mu ^{{2 \mathord{\left/
 {\vphantom {2 \alpha }} \right.
 \kern-\nulldelimiterspace} \alpha }}}} \right)}\nonumber\\
  &\pi \Gamma \left( {1 + {2 \mathord{\left/
 {\vphantom {2 \alpha }} \right.
 \kern-\nulldelimiterspace} \alpha }} \right)\Gamma \left( {1 - {2 \mathord{\left/
 {\vphantom {2 \alpha }} \right.
 \kern-\nulldelimiterspace} \alpha }} \right){{{\left( {{\gamma _{th}}} \right)}^{\frac{2}{\alpha }}}\prod\limits_{j = 0}^{l - 1} {\left( {{2 \mathord{\left/
 {\vphantom {2 \alpha }} \right.
 \kern-\nulldelimiterspace} \alpha }- j} \right){r^{2 - l\alpha }}} }\Big]^{m_l}.
 \end{align}

Substituting \eqref{CDFAP_D6} into \eqref{CDFAP_D1}, we derive the CDF of $\gamma_{ap}$ in \eqref{CDFAP}.

\section{  Proof of Lemma 2}\label{aplemma2}

From \eqref{SINR_Sensor_Eve}, the CDF of $\gamma_{s,e}$ is given by
\begin{align}
{F_{{\gamma _{s,e}}}}\left( {{\gamma _{th}}} \right) & =\Pr \Bigg\{ {\mathop {\max }\limits_{{e_k} \in {\Phi _{s,e}}} \Bigg\{ {\frac{{{{\left| {{h_{{s_0},{e_k}}}} \right|}^2}{{\left| {{X_{{s_0},{e_k}}}} \right|}^{ - \alpha }}}}{{I{n_{s,e}} + {{{\delta ^2}} \mathord{\left/
 {\vphantom {{{\delta ^2}} {{P_s}}}} \right.
 \kern-\nulldelimiterspace} {{P_s}}}}}} \Bigg\} \le {\gamma _{th}}} \Bigg\}
 \nonumber \\
 & \hspace{-1.6 cm}  \mathop  = \limits^{\left( a \right)} \exp \left\{ { - \lambda _e^s{\smallint _{{R^2}}}{e^{ - {{{\delta ^2}{\gamma _{th}}{{\left| {{X_{{s_0},{e_k}}}} \right|}^\alpha }} \mathord{\left/
 {\vphantom {{{\delta ^2}{\gamma _{th}}{{\left| {{X_{{s_0},{e_k}}}} \right|}^\alpha }} {{P_s}}}} \right.
 \kern-\nulldelimiterspace} {{P_s}}}}}{{\cal L}_{I{n_{s,e}}}}\left( {{\gamma _{th}}{{\left| {{X_{{s_0},{e_k}}}} \right|}^\alpha }} \right)} \right.
  \nonumber \\ 
  &  \hspace{-1.2cm} \left. {d\left| {{X_{{s_0},{e_k}}}} \right|} \right\} 
 \nonumber \\
 & \hspace{-1.6 cm}\mathop  = \limits^{\left( b \right)} \exp \left\{ { - 2\pi \lambda _e^s\int_0^\infty  {{e^{ - {{{\delta ^2}{\gamma _{th}}{r^\alpha }} \mathord{\left/
 {\vphantom {{{\delta ^2}{\gamma _{th}}{r^\alpha }} {{P_s}}}} \right.
 \kern-\nulldelimiterspace} {{P_s}}}}}{\mathcal{L}_{I{n_{s,e}}}}\left( {{\gamma _{th}}{r^\alpha }} \right)rdr} } \right\},
\end{align}
where $(a)$ follows from the generating functionnal of HPPP in \cite{stoyanstochastic}, $(b)$ is obtained by converting cartesian coordinates to polar coordinates.

Using the generating functionnal of HPPP in \cite{stoyanstochastic}, ${\left| {{h_{i,{e_k}}}} \right|^2} \sim \exp \left( 1 \right)$, and $H_j^{ap,e} = {\left| {{{\bf{h}}_{j,{e_k}}}\frac{{{{\bf{h}}_{j,s{k_j}}}^\dag }}{{\left\| {{{\bf{h}}_{j,s{k_j}}}} \right\|}}} \right|^2} \sim \exp \left( 1 \right)$, we derive the Laplace transform of $I_{s,e}$ and $I_{ap,e}$ as
 \begin{align}\label{LAP_se}
&{\mathcal{L}_{{I_{s,e}}}}\left( s \right)  \nonumber \\
 &\hspace{0.3cm} =\exp \left( { - \smallint \left[ {1 - {\mathbbm{E}_h}\left( {\exp \left( { - s{{\left| {{h_{i,{e_k}}}} \right|}^2}{y^{ - \alpha }}} \right)} \right)} \right]{\lambda _s}{\rho _s}2\pi ydy} \right)
 \nonumber \\
 &\hspace{0.3cm}    = \exp \left\{ { - {\lambda _s}{\rho _s}\pi \Gamma \left( {1 + {2 \mathord{\left/
 {\vphantom {2 \alpha }} \right.
 \kern-\nulldelimiterspace} \alpha }} \right)\Gamma \left( {1 - {2 \mathord{\left/
 {\vphantom {2 \alpha }} \right.
 \kern-\nulldelimiterspace} \alpha }} \right){s^{{2 \mathord{\left/
 {\vphantom {2 \alpha }} \right.
 \kern-\nulldelimiterspace} \alpha }}}} \right\},
\end{align}
and
\begin{align}\label{LAP_ape}
&{\mathcal{L}_{{I_{ap,e}}}}\left( s \right)  \nonumber \\
 &\hspace{0.3cm} =\exp \left( { - \smallint \left[ {1 - {\mathbbm{E}_h}\left( {\exp \left( { - s\mu H_j^{ap,e}{y^{ - \alpha }}} \right)} \right)} \right]{\lambda _{ap}}{\rho _{ap}}2\pi ydy} \right)
 \nonumber \\
 &\hspace{0.3cm} = \exp \left\{ { - {\lambda _{ap}}{\rho _{ap}}\pi {\mu ^{\frac{2}{\alpha }}}\Gamma \left( {1 + {2 \mathord{\left/
 {\vphantom {2 \alpha }} \right.
 \kern-\nulldelimiterspace} \alpha }} \right)\Gamma \left( {1 - {2 \mathord{\left/
 {\vphantom {2 \alpha }} \right.
 \kern-\nulldelimiterspace} \alpha }} \right){s^{{2 \mathord{\left/
 {\vphantom {2 \alpha }} \right.
 \kern-\nulldelimiterspace} \alpha }}}} \right\},
\end{align}
respectively.

With the Laplace transform of ${I_{s,e}}$ and ${I_{ap,e}}$, we derive the Laplace transform of $In_{s,e}$ as
\begin{align}\label{LAP_inap}
{\mathcal{L}_{I{n_{s,e}}}}\left( s \right) =&\exp \left\{ { - {\lambda _s}{\rho _s}\pi \Gamma \left( {1 + {2 \mathord{\left/
 {\vphantom {2 \alpha }} \right.
 \kern-\nulldelimiterspace} \alpha }} \right)\Gamma \left( {1 - {2 \mathord{\left/
 {\vphantom {2 \alpha }} \right.
 \kern-\nulldelimiterspace} \alpha }} \right){s^{{2 \mathord{\left/
 {\vphantom {2 \alpha }} \right.
 \kern-\nulldelimiterspace} \alpha }}} - {\lambda _{ap}}} \right.
\nonumber \\ 
 & \hspace{-1cm} \left. {{\rho _{ap}}\pi {\mu ^{{2 \mathord{\left/
 {\vphantom {2 \alpha }} \right.
 \kern-\nulldelimiterspace} \alpha }}}\Gamma \left( {1 + {2 \mathord{\left/
 {\vphantom {2 \alpha }} \right.
 \kern-\nulldelimiterspace} \alpha }} \right)\Gamma \left( {1 - {2 \mathord{\left/
 {\vphantom {2 \alpha }} \right.
 \kern-\nulldelimiterspace} \alpha }} \right){s^{{2 \mathord{\left/
 {\vphantom {2 \alpha }} \right.
 \kern-\nulldelimiterspace} \alpha }}}} \right\}.
\end{align}

Substituting \eqref{LAP_inap} into \eqref{CDFse_D1}, we derive the CDF of $\gamma_{s,e}$ in \eqref{CDFAP_E}.

\section{  Proof of Lemma 3}\label{aplemma3}
From \eqref{SINR_sink}, the CDF of ${\gamma _{sk}}$ is given by
\begin{align}\label{CDFSK_D1}
{F_{{\gamma _{sk}}}}\left( {{\gamma _{th}}} \right)
 &= \int_0^\infty  {\Pr \Bigg[ {\frac{{{{\left\| {{{\bf{g}}_{a{p_0},S{k_0}}}} \right\|}^2}{r^{ - \beta }}}}{{I{n_{ap,sk}} + {{{\delta ^2}} \mathord{\left/
 {\vphantom {{{\delta ^2}} {{P_{ap}}}}} \right.
 \kern-\nulldelimiterspace} {{P_{ap}}}}}} \le {\gamma _{th}}} \Bigg]} 2\pi {\lambda _{sk}}r
 \nonumber \\
 & \hspace{0.4cm}  \exp \left( { - \pi {\lambda _{sk}}{r^2}} \right)dr.
\end{align}

The CDF of the sink SINR at distance $r$ from its corresponding access point  is derived as
\begin{align}\label{CDFSK_D2}
\Pr \Bigg[ {\frac{{{{\left\| {{{\bf{g}}_{a{p_0},s{k_0}}}} \right\|}^2}{r^{ - \beta }}}}{{I{n_{ap,sk}} + {{{\delta ^2}} \mathord{\left/
 {\vphantom {{{\delta ^2}} {{P_{ap}}}}} \right.
 \kern-\nulldelimiterspace} {{P_{ap}}}}}} \le {\gamma _{th}}} \Bigg]&
  \nonumber \\
 & \hspace{-4.3cm}    = 1 - \sum\limits_{m = 0}^{M - 1} {\frac{1}{{m!}}{\mathbbm{E}_{{\Phi _{ap,a}}}}} \Big\{ {\int_0^\infty  {{{\left[ {{\gamma _{th}}{r^\beta }\left( {\tau  + {{{\delta ^2}} \mathord{\left/
 {\vphantom {{{\delta ^2}} {{P_{ap}}}}} \right.
 \kern-\nulldelimiterspace} {{P_{ap}}}}} \right)} \right]}^m}} }
  \nonumber \\
 & \hspace{-3.6cm} {\exp \left[ { - {\gamma _{th}}{r^\beta }\left( {\tau  + {{{\delta ^2}} \mathord{\left/
 {\vphantom {{{\delta ^2}} {{P_{ap}}}}} \right.
 \kern-\nulldelimiterspace} {{P_{ap}}}}} \right)} \right]d\Pr \left( {I{n_{ap,sk}} \le \tau } \right)} \Big\}.
\end{align}

Note that ${\left( { - \left( {\tau  + {{{\delta ^2}} \mathord{\left/
 {\vphantom {{{\delta ^2}} {{P_{ap}}}}} \right.
 \kern-\nulldelimiterspace} {{P_{ap}}}}} \right){\gamma _{th}}} \right)^m}{e^{ - \left( {\tau  + {{{\delta ^2}} \mathord{\left/
 {\vphantom {{{\delta ^2}} {{P_{ap}}}}} \right.
 \kern-\nulldelimiterspace} {{P_{ap}}}}} \right)\gamma _{th}^{\left\{ s \right\}}{r^\beta }}}={\left. {\frac{{{d^m}\left( {{e^{ - {\gamma _{th}}x\left( {\tau  + {{{\delta ^2}} \mathord{\left/
 {\vphantom {{{\delta ^2}} {{P_{ap}}}}} \right.
 \kern-\nulldelimiterspace} {{P_{ap}}}}} \right)}}} \right)}}{{d{x^m}}}} \right|_{x = {r^\beta }}} $, we rewrite \eqref{CDFSK_D2} as
\begin{align}\label{CDFSK_D3}
\Pr \Bigg[ {\frac{{{{\left\| {{{\bf{g}}_{a{p_0},s{k_0}}}} \right\|}^2}{r^{ - \beta }}}}{{I{n_{ap,sk}} + {{{\delta ^2}} \mathord{\left/
 {\vphantom {{{\delta ^2}} {{P_{ap}}}}} \right.
 \kern-\nulldelimiterspace} {{P_{ap}}}}}} \le {\gamma _{th}}} \Bigg] =&  1 - {\mathbbm{E}_{{\Phi _{ap,a}}}}
 \nonumber \\
 & \hspace{-4.6cm} \left\{ {\int_0^\infty  {\exp \left[ { - {\gamma _{th}}{r^\beta }\left( {\tau  + {{{\delta ^2}} \mathord{\left/
 {\vphantom {{{\delta ^2}} {{P_{ap}}}}} \right.
 \kern-\nulldelimiterspace} {{P_{ap}}}}} \right)} \right]d\Pr \left( {I{n_{ap,sk}} \le \tau } \right)} } \right\}
  \nonumber \\
 & \hspace{-4.6cm}  - \sum\limits_{m = 1}^{M - 1} {\frac{{{{\left( {{r^\beta }} \right)}^m}}}{{m!{{\left( { - 1} \right)}^m}}}{\mathbbm{E}_{{\Phi _{ap,a}}}}} \left\{ {\int_0^\infty  {{{\left. {\frac{{{d^m}\left( {{e^{ - {\gamma _{th}}x\left( {\tau  + {{{\delta ^2}} \mathord{\left/
 {\vphantom {{{\delta ^2}} {{P_{ap}}}}} \right.
 \kern-\nulldelimiterspace} {{P_{ap}}}}} \right)}}} \right)}}{{d{x^m}}}} \right|}_{x = {r^\beta }}}} } \right.
 \nonumber \\
 & \hspace{-4.6cm}
 {d\Pr \left( {I{n_{ap,sk}} \le \tau } \right)} \Bigg\}
 \nonumber \\
 & \hspace{-5cm}
 = 1 - \exp \left( { - {{{\gamma _{th}}{r^\beta }{\delta ^2}} \mathord{\left/
 {\vphantom {{{\gamma _{th}}{r^\beta }{\delta ^2}} {{P_{ap}}}}} \right.
 \kern-\nulldelimiterspace} {{P_{ap}}}}} \right){\mathcal{L}_{I{n_{ap,sk}}}}\left( {{\gamma _{th}}{r^\beta }} \right) - \sum\limits_{m = 1}^{M - 1} {\frac{{{{\left( {{r^\beta }} \right)}^m}}}{{m!{{\left( { - 1} \right)}^m}}}}
  \nonumber \\
 & \hspace{-4.6cm} {\left. {\frac{{{d^m}\left( {\exp \left( { - {{{\gamma _{th}}x{\delta ^2}} \mathord{\left/
 {\vphantom {{{\gamma _{th}}x{\delta ^2}} {{P_{ap}}}}} \right.
 \kern-\nulldelimiterspace} {{P_{ap}}}}} \right){\mathcal{L}_{I{n_{ap,sk}}}}\left( {{\gamma _{th}}x} \right)} \right)}}{{d{x^m}}}} \right|_{x = {r^\beta }}} .
\end{align}
Since $I{n_{ap,sk}} = {\sum _{j \in {\Phi _{ap,a}}\backslash \left\{ {a{p_0}} \right\}}}{\left| {{{\bf{g}}_{j,s{k_0}}}\frac{{{{\bf{h}}_{j,s{k_j}}}^\dag }}{{\left\| {{{\bf{h}}_{j,s{k_j}}}} \right\|}}} \right|^2}{\left| {{X_{j,s{k_0}}}} \right|^{ - \beta }}$, using the generating functionnal of HPPP and  ${\left| {{{\bf{g}}_{j,S{k_0}}}\frac{{{{\bf{h}}_{j,s{k_j}}}^\dag }}{{\left\| {{{\bf{h}}_{j,s{k_j}}}} \right\|}}} \right|^2} \sim \exp \left( 1 \right)$ , we derive the Laplace transform of $I{n_{ap,sk}}$ as
\begin{align}\label{LAP_inapsk}
{\mathcal{L}_{I{n_{ap,sk}}}}\left( s \right) =\exp \left\{ { - {\lambda _{ap}}{\rho _{ap}}\pi \Gamma \left( {1 + {2 \mathord{\left/
 {\vphantom {2 \beta }} \right.
 \kern-\nulldelimiterspace} \beta }} \right)\Gamma \left( {1 - {2 \mathord{\left/
 {\vphantom {2 \beta }} \right.
 \kern-\nulldelimiterspace} \beta }} \right){s^{{2 \mathord{\left/
 {\vphantom {2 \beta }} \right.
 \kern-\nulldelimiterspace} \beta }}}} \right\}.
\end{align}

Substituting \eqref{LAP_inapsk} into \eqref{CDFSK_D3}, we obtain
\begin{align}\label{CDFSK_D4}
&\Pr \Bigg[ {\frac{{{{\left\| {{{\bf{g}}_{a{p_0},s{k_0}}}} \right\|}^2}{r^{ - \beta }}}}{{I{n_{ap,sk}} + {{{\delta ^2}} \mathord{\left/
 {\vphantom {{{\delta ^2}} {{P_{ap}}}}} \right.
 \kern-\nulldelimiterspace} {{P_{ap}}}}}} \le {\gamma _{th}}} \Bigg]=1 - \exp \Big\{ { - {\lambda _{ap}}{\rho _{ap}}\pi}\nonumber\\
& \Gamma \left( {1 + {2 \mathord{\left/
 {\vphantom {2 \beta }} \right.
 \kern-\nulldelimiterspace} \beta }} \right) {\Gamma \left( {1 - {2 \mathord{\left/
 {\vphantom {2 \beta }} \right.
 \kern-\nulldelimiterspace} \beta }} \right){{\left( {{\gamma _{th}}} \right)}^{{2 \mathord{\left/
 {\vphantom {2 \beta }} \right.
 \kern-\nulldelimiterspace} \beta }}}{r^2} - {{{\gamma _{th}}{r^\beta }{\delta ^2}} \mathord{\left/
 {\vphantom {{{\gamma _{th}}{r^\beta }{\delta ^2}} {{P_{ap}}}}} \right.
 \kern-\nulldelimiterspace} {{P_{ap}}}}} \Big\}\nonumber\\
& - \sum\limits_{m = 1}^{M - 1} {\frac{{{{\left( {{r^\beta }} \right)}^m}}}{{m!{{\left( { - 1} \right)}^m}}}{{\left. {\frac{{{d^m}\left( {U\left( x \right)} \right)}}{{d{x^m}}}} \right|}_{x = {r^\beta }}}}
\end{align}
with $U\left( x \right) =  \exp\Big\{ { - {\lambda _{ap}}{\rho _{ap}}\pi \Gamma \left( {1 + {2 \mathord{\left/
 {\vphantom {2 \beta }} \right.
 \kern-\nulldelimiterspace} \beta }} \right)\Gamma \left( {1 - {2 \mathord{\left/
 {\vphantom {2 \beta }} \right.
 \kern-\nulldelimiterspace} \beta }} \right){{\left( {{\gamma _{th}}x} \right)}^{{2 \mathord{\left/
 {\vphantom {2 \beta }} \right.
 \kern-\nulldelimiterspace} \beta }}}}$\\$
- {{{\gamma _{th}}x{\delta ^2}} \mathord{\left/
 {\vphantom {{{\gamma _{th}}x{\delta ^2}} {{P_{ap}}}}} \right.
 \kern-\nulldelimiterspace} {{P_{ap}}}} \Big\}$.

 We then apply the Fa$\grave{a}$
 di Bruno's formula to solve the derivative of $m$th order as follows:
\begin{align}\label{CDFSK_D6}
&{\left. {\frac{{{d^m}\left[ {\exp \left( {U\left( x \right)} \right)} \right]}}{{d{x^m}}}} \right|_{x = {r^\beta }}} ={\sum {\frac{1}{{\prod\limits_{l = 1}^m {{m_l}!l{!^{{m_l}}}} }}} }\exp\Big\{{ - {\lambda _{ap}}{\rho _{ap}}\pi }\nonumber\\
& {\Gamma \left( {1 + {2 \mathord{\left/
 {\vphantom {2 \beta }} \right.
 \kern-\nulldelimiterspace} \beta }} \right)\Gamma \left( {1 - {2 \mathord{\left/
 {\vphantom {2 \beta }} \right.
 \kern-\nulldelimiterspace} \beta }} \right){{\left( {{\gamma _{th}}} \right)}^{{2 \mathord{\left/
 {\vphantom {2 \beta }} \right.
 \kern-\nulldelimiterspace} \beta }}}{r^2} - {{{\gamma _{th}}{r^\beta }{\delta ^2}} \mathord{\left/
 {\vphantom {{{\gamma _{th}}{r^\beta }{\delta ^2}} {{P_{ap}}}}} \right.
 \kern-\nulldelimiterspace} {{P_{ap}}}}}\Big\}\Big[  - {\lambda _{ap}} \nonumber\\
&{\rho _{ap}}\pi\frac{2}{\beta }\Gamma \left( {1 + {2 \mathord{\left/
 {\vphantom {2 \beta }} \right.
 \kern-\nulldelimiterspace} \beta }} \right)\Gamma \left( {1 - {2 \mathord{\left/
 {\vphantom {2 \beta }} \right.
 \kern-\nulldelimiterspace} \beta }} \right){{\left( {{\gamma _{th}}} \right)}^{{2 \mathord{\left/
 {\vphantom {2 \beta }} \right.
 \kern-\nulldelimiterspace} \beta }}}{x^{{2 \mathord{\left/
 {\vphantom {2 \beta }} \right.
 \kern-\nulldelimiterspace} \beta } - 1}} - {{{\gamma _{th}}{\delta ^2}} \mathord{\left/
 {\vphantom {{{\gamma _{th}}{\delta ^2}} {{P_{ap}}}}} \right.
 \kern-\nulldelimiterspace} {{P_{ap}}}}\Big]^{{m_1}}\nonumber\\
 & \prod\limits_{l = 2}^m \Big[ - {\lambda _{ap}}{\rho _{ap}}\pi \Gamma \left( {1 + {2 \mathord{\left/
 {\vphantom {2 \beta }} \right.
 \kern-\nulldelimiterspace} \beta }} \right)\Gamma \left( {1 - {2 \mathord{\left/
 {\vphantom {2 \beta }} \right.
 \kern-\nulldelimiterspace} \beta }} \right){{\left( {{\gamma _{th}}} \right)}^{{2 \mathord{\left/
 {\vphantom {2 \beta }} \right.
 \kern-\nulldelimiterspace} \beta }}} \prod\limits_{j = 0}^{l - 1} \left( {{2 \mathord{\left/
 {\vphantom {2 \beta }} \right.
 \kern-\nulldelimiterspace} \beta } - j} \right) \Big.
 \nonumber\\
 &
 \Bigg.
 {{x^{{2 \mathord{\left/
 {\vphantom {2 \beta }} \right.
 \kern-\nulldelimiterspace} \beta } - l}}}  \Big]^{{m_l}}.
 \end{align}

Based on \eqref{CDFSK_D6}, \eqref{CDFSK_D4}, and \eqref{CDFSK_D1}, we derive the CDF of $\gamma_{sk}$ in \eqref{CDFSK_Sink_2hop}.

\section{  Proof of Lemma 4}\label{aplemma4}

From \eqref{SINR_AP_EVE}, the CDF of $\gamma_{ap,e}$ is given by
\begin{align}\label{CDFse_D1}
{F_{{\gamma _{s,e}}}}\left( {{\gamma _{th}}} \right) & ={\mathbbm{E}_{{\Phi _{ap,a}}}}\Bigg\{ {{\mathbbm{E}_{{\Phi _{ap,e}}}}\Bigg\{ {\prod\limits_{e{\Phi _{ap,e}}} {\Pr \left\{ {\frac{{{{\left| {{g_{a{p_0},{e_k}}}} \right|}^2}}}{{I{n_{ap,e}} + {{{\sigma ^2}} \mathord{\left/
 {\vphantom {{{\sigma ^2}} {{P_{ap}}}}} \right.
 \kern-\nulldelimiterspace} {{P_{ap}}}}}}} \right.} } }
\nonumber \\
  &  \hspace{-1.2cm}   { {\Bigg. {{{\left| {{X_{a{p_0},{e_k}}}} \right|}^{ - \beta }} \le {\gamma _{th}}} \Bigg|{\Phi _{ap,a}},{\Phi _{ap,e}}} \Bigg\}} \Bigg\}
 \nonumber \\
 & \hspace{-1.6cm}   \mathop  = \limits^{\left( a \right)}\exp \Big\{ { - \lambda _e^{ap}\int_{{R^2}} {{e^{ - {{{\sigma ^2}{\gamma _{th}}{{\left| {{X_{a{p_0},{e_k}}}} \right|}^\beta }} \mathord{\left/
 {\vphantom {{{\sigma ^2}{\gamma _{th}}{{\left| {{X_{a{p_0},{e_k}}}} \right|}^\beta }} {{P_{ap}}}}} \right.
 \kern-\nulldelimiterspace} {{P_{ap}}}}}}} }
 \nonumber \\
 & \hspace{-1.2cm}    {{\mathcal{L}_{I{n_{ap,e}}}}\Big( {{\gamma _{th}}{{\left| {{X_{a{p_0},{e_k}}}} \right|}^\beta }} \Big)de} \Big\}
 \nonumber \\
 & \hspace{-1.6cm}
\mathop  = \limits^{\left( b \right)} \exp \Big\{ { - 2\pi \lambda _e^{ap}\int_0^\infty  {{e^{ - {{{\sigma ^2}{\gamma _{th}}{r^\beta }} \mathord{\left/
 {\vphantom {{{\sigma ^2}{\gamma _{th}}{r^\beta }} {{P_{ap}}}}} \right.
 \kern-\nulldelimiterspace} {{P_{ap}}}}}}{\mathcal{L}_{I{n_{ap,e}}}}\left( {{\gamma _{th}}{r^\beta }} \right)rdr} } \Big\},
\end{align}
where $(a)$ follows from the generating functionnal of HPPP in \cite{stoyanstochastic}, $(b)$ is obtained by converting cartesian coordinates to polar coordinates.

Using the generating functionnal of HPPP in \cite{stoyanstochastic}, we derive the Laplace transform of $I_{ap,e}$ as
 \begin{align}\label{LAP_ape}
{\mathcal{L}_{{I_{ap,e}}}}\left( s \right)  = \exp \left\{ { - {\lambda _{ap}}{\rho _{ap}}\pi \Gamma \left( {1 + {2 \mathord{\left/
 {\vphantom {2 \beta }} \right.
 \kern-\nulldelimiterspace} \beta }} \right)\Gamma \left( {1 - {2 \mathord{\left/
 {\vphantom {2 \beta }} \right.
 \kern-\nulldelimiterspace} \beta }} \right){s^{{2 \mathord{\left/
 {\vphantom {2 \beta }} \right.
 \kern-\nulldelimiterspace} \beta }}}} \right\}.
\end{align}

Plugging \eqref{LAP_ape} into \eqref{CDFse_D1}, we derive the CDF of $\gamma_{s,e}$ in \eqref{CDFSK_E}.

\bibliographystyle{IEEEtran}
\bibliography{mybib}

\balance
\end{document}